\documentclass[12pt,a4paper]{article}
\usepackage[utf8]{inputenc}
\usepackage{lmodern}
\usepackage[T1]{fontenc}
\usepackage[english]{babel}
\usepackage{microtype}
\usepackage{amsmath,amssymb,mathtools}
\usepackage{amsthm}
\usepackage{graphicx}
\usepackage[natbibapa]{apacite}
\usepackage{hyperref}
\usepackage{booktabs}
\usepackage{array}
\usepackage{enumitem}

\newcommand{\R}{\mathbb{R}}
\newcommand{\Simp}{\Sigma_2}   

\newtheorem{theorem}{Theorem}[section]
\newtheorem{lemma}[theorem]{Lemma}
\newtheorem{corollary}[theorem]{Corollary}
\newtheorem{proposition}[theorem]{Proposition}
\theoremstyle{definition}
\newtheorem{definition}[theorem]{Definition}
\theoremstyle{remark}
\newtheorem{remark}[theorem]{Remark}

\begin{document}

\title{Crisis, Disengagement, and Structural Realignment: A Threshold Model of Radical-Party Support}

\author{%
  Alexander Omelchenko\thanks{%
    Constructor University Bremen gGmbH, Campus Ring~1, 28759 Bremen,
    Germany. Email: \texttt{aomelchenko@constructor.university}.}%
}
\date{\today}

\maketitle

\begin{abstract}
When does a crisis-induced surge in radical-party support fade away, and
when does it become a durable realignment? We address this in a
mathematical sociology threshold model on a conserved population. The
baseline admits a global classification through a Perron--Frobenius
threshold. Adding a crisis-induced disengagement compartment, we
separate state shocks (which alter the current state) from structural
shocks (which alter parameters). State shocks affect transients but
cannot move the long-run attractor; durable realignment requires
structural threshold-crossing. We derive a critical shock amplitude and
a finite mobilisation-window bound, and show that cumulative structural
shifts can produce staircase realignment. A stylised illustration uses
German federal elections, 2013--2025.
\end{abstract}

\begin{quote}
\small\noindent\textbf{Keywords:} voter dynamics, political polarisation,
radical-party support, threshold models, simplex dynamics, bifurcation,
compartmental models
\end{quote}

\section{Introduction}

A central question in the analysis of post-crisis European politics is
when a surge in radical-party support reflects a transient electoral
shock and when it reflects a durable political realignment. In dynamical
terms, the difference is between a shock to the state of the system and a
shock to the vector field governing its subsequent evolution. The former
changes the current distribution of voters; the latter changes the
post-crisis recruitment and reabsorption environment. The two mechanisms
are not empirically distinguishable from a single election: both may
produce a sharp rise in radical vote share and both can be politically
consequential at their peak. This article develops a conserved-electorate
voter-flow model that separates the two cases analytically and derives
explicit thresholds for transient surges and long-run regime shifts.

Throughout the paper, ``radical'' is used as a modelling category rather
than as a legal or normative classification. It denotes party blocs outside
mainstream party competition that recruit from the mainstream electorate and
may trigger reactive mobilisation by an opposed bloc. The model is therefore
agnostic about the detailed ideological content of the blocs; empirical
applications require a separate mapping from parties to model compartments.

The empirical motivation is well known. Across Western democracies,
challenger and radical parties have grown while mainstream parties have
weakened, especially in the aftermath of successive crises
\citep{mudde2019,norris2019,kriesi2008,hobolt2016}. The qualitative
feature of interest is not simply that radical support rises during
crises, but that it may settle at a higher post-crisis reference level
afterwards. The stylised interpretation developed below treats such a
ratchet pattern as the possible cumulative effect of structural shifts
in the recruitment environment, rather than as repeated temporary
volatility alone. The aim of the present paper is to analyse this
distinction---between transient electoral surges and durable
realignment---in a minimal mathematical sociology framework, not to
explain the trajectory of any particular party.

The framework belongs to the tradition of threshold models in
mathematical sociology \citep{granovetter1978,macy1991}, in which
aggregate macro-level outcomes emerge from population-level rules of
recruitment, conversion, and reabsorption. We adapt this tradition to a
simplex setting with explicit competition between radical wings and
reactive polarisation, which together generate the threshold structure
analysed below.

Existing formal models of voter flows do not directly isolate the distinction
at issue here: the difference between temporary electoral surges and durable
regime shifts in radical support. Bounded-confidence and network models
describe individual opinion formation and cluster dynamics
\citep{deffuant2000,hegselmann2002,lorenz2007,baumann2020,flache2017}, but
they do not directly address the aggregate vote-share question studied
here: whether a surge in radical support is a transient displacement of
the state or evidence of a change in the long-run regime selected by the
vector field.

Aggregate ODE models closer to the present
paper include the satisficing model of \citet{yang2020}, which tracks party
\emph{positions} rather than voter shares, and the sociophysics model of
\citet{diep2024} applied to the 2024 US election. The most directly comparable
line of work is the cRUD framework of \citet{volkening2020}, later extended by
\citet{branstetter2026}: Democratic, Republican, and undecided fractions evolve
on a conserved-electorate simplex, with parameters fitted to polling data and
probabilistic forecasts produced by simulation. This framework shares our
simplex constraint and conserved electorate, but is designed for election
forecasting rather than structural analysis. It asks who wins a given race;
the present paper asks whether the system has a stable mainstream equilibrium,
and under what conditions a crisis can permanently displace it. The model also
targets multi-party European systems, incorporating \emph{reactive
polarisation} between wings: growth of one radical bloc amplifies recruitment
by the opposite wing, a feedback absent from two-party models.

The mathematical apparatus we use---a Perron--Frobenius spectral threshold
governing global stability of a radical-free equilibrium---is structurally
analogous to the next-generation matrix framework for basic reproduction
numbers in epidemiology \citep{diekmann1990,vandendriessche2002}. The
substantive interpretation differs: instead of secondary infections per
infected individual, $\mathcal{R}_{\mathrm{base}}$ measures whether
recruitment plus reactive polarisation are strong enough relative to
reabsorption to sustain a positive long-run radical share. The connection
is conceptual rather than mechanical---politics is not contagion---but it
makes available a well-developed body of techniques for global threshold
analysis on a simplex.

A related literature studies democratic erosion and institutional
backsliding \citep{levitsky2018,haggard2021,grillo2024}. The model below
is connected to that literature in motivation but differs in object and
method. Crisis-induced disengagement is represented as a state variable,
while persistent changes in the recruitment and reabsorption environment
are represented as structural parameter shifts. The model therefore does
not explain institutional erosion directly; instead, it asks how changes
in the political environment affect aggregate voter flows and long-run
radical-party support.

To capture the transient amplification component of crises, the model
introduces a compartment for voters who temporarily disengage from active
party competition during a crisis. This compartment is not total
abstention: it represents crisis-induced disengagement above ordinary
background non-voting. Survey evidence documents politically dissatisfied,
weakly attached, and non-participating voters
\citep{dalton2004,torcal2006,eurobarometer2013,gles2021}. The model does
not assume that disengagement mechanically precedes radical support in
every case. Rather, it treats disengagement as a transient reservoir from
which voters may either return to the mainstream or be mobilised by
radical parties. The branching between these two outflows governs the
size and duration of post-crisis radical surges, but not the long-run
equilibrium itself: durable realignment arises from structural changes in
the recruitment environment.

The paper makes three contributions, organised around a single
state--structure decomposition.

First, within the conserved-electorate voter-flow model, we separate two
ways in which a crisis can enter the dynamics. A state shock changes the
current distribution of voters, for example by moving part of the
mainstream pool into crisis-induced disengagement. For fixed post-shock
parameters, such a shock can affect transients but cannot change the
long-run attractor. A structural shock changes the post-crisis parameter
vector and can therefore change the attracting equilibrium. The two
mechanisms have different thresholds: in the symmetric reduction,
state-driven transient growth requires a critical shock amplitude
$\Delta_c$, while lasting realignment requires crossing the spectral
threshold $\lambda_{\mathrm{PF}}(M^{-1}K)=1$.

Second, we develop the global theory behind this decomposition. The
baseline three-group simplex system admits a complete global
classification governed by the Perron--Frobenius root
$\mathcal{R}_{\mathrm{base}}=\lambda_{\mathrm{PF}}(M^{-1}K)$, with no
periodic orbits and no attracting bistability. Adding a disengagement
compartment preserves the equilibrium set on the face $A=0$ but
introduces a transient-amplification mechanism. In the symmetric
reduction, the total radical share $S=L+R$ satisfies a closed
two-dimensional system with an explicit critical state-shock amplitude
and mobilisation-window bound.

Third, we show how staircase dynamics arise from cumulative structural
change. Pure state shocks cannot generate a sequence of new long-run
floors under fixed parameters. By contrast, individually insufficient
structural shifts can cumulate until the spectral threshold is crossed;
after crossing, further positive structural shifts lower the long-run
mainstream share and raise the radical-support floor. The empirical
implication is that temporary crisis volatility and durable political
realignment should be modelled as distinct mechanisms.

Substantively, the model yields three qualitative implications. First,
temporary disengagement shocks can produce large but reversible surges in
radical-party support, especially when re-engagement is slow. Second, a
durable post-crisis floor requires a structural change in the recruitment
and reabsorption environment, not merely a large temporary displacement
of voters. Third, repeated crises can generate a staircase pattern only
when their structural components accumulate across episodes. These
implications are qualitative: the model is intended to clarify
mechanisms and thresholds, not to provide an election-forecasting
device.

Substantively, the model yields four qualitative implications. First, in
a subcritical fixed-parameter regime, a state shock can generate a
transient radical surge, but the return speed should increase with
reabsorption and re-engagement rates and the surge duration should
increase with the size of the disengagement reservoir. Second,
persistent post-crisis radical support requires post-shock structural
parameters above the recruitment threshold. Third, repeated crises
without structural accumulation should not create a sequence of
successively higher long-run radical-support equilibria. Fourth,
positive cross-wing couplings imply counter-mobilisation: growth in one
radical bloc increases recruitment pressure into the opposed bloc.
These implications are qualitative: the model clarifies mechanisms and
thresholds rather than providing an election-forecasting device.

The paper is organised as follows. Section~\ref{sec:model} introduces the
baseline voter-flow model. Section~\ref{sec:baseline} analyses the baseline system and establishes
path-independence of the fixed-parameter dynamics. Section~\ref{sec:extended}
extends the model by adding disengagement and distinguishing state shocks from
structural shocks. Section~\ref{sec:legacy} derives the conditions under which
crises leave a lasting political legacy, including the critical shock amplitude
and the structural threshold for permanent equilibrium shifts. Section~\ref{sec:germany}
offers a stylised empirical illustration using German federal elections,
2013--2025. Section~\ref{sec:discussion} discusses implications and
limitations. Proofs and additional numerical illustrations are collected
in the appendices.

\section{A Voter-Flow Model of Political Competition}\label{sec:model}

We begin with a baseline model in which the electorate is distributed across
three groups: a left-radical bloc $L(t)$, a right-radical bloc $R(t)$, and a
mainstream bloc $C(t)$ comprising all remaining party support. All variables
are normalised shares of a fixed reference electorate, so that
\[
  L(t)\ge0,\qquad R(t)\ge0,\qquad C(t)=1-L(t)-R(t)\ge0,
\]
and the feasible state space is the probability simplex
\[
  \Simp = \{(L,R)\in\R^2 : L\ge0,\;R\ge0,\;L+R\le1\}.
\]
The conserved-electorate constraint is a deliberate modelling choice: the
model studies redistribution within a fixed political population rather than
changes in population size. In the baseline, ordinary background non-voting
is not tracked separately: $C$ collects mainstream-party support together
with the ordinary abstention level. Crisis-induced disengagement above this
background is introduced as a separate compartment $A(t)$ in
Section~\ref{sec:extended}; the empirical mapping in
Section~\ref{sec:germany} accordingly distinguishes $C_{\mathrm{obs}}$ from
observed non-voting $N$.

In sociological terms, the compartments correspond to political-attachment
categories: $L$ and $R$ to active radical attachment on opposing
ideological poles, $C$ to mainstream/non-radical attachment, and (in the
extended model of Section~\ref{sec:extended}) $A$ to crisis-induced
detachment from political participation. The recruitment, reactive
polarisation, and reabsorption parameters represent socialisation,
counter-mobilisation, and re-engagement processes operating at the
aggregate level. The model is therefore neutral about the micro-level
mechanisms producing these flows; it specifies their aggregate
consequences.

The dynamics are governed by
\begin{equation}\label{eq:model}
\begin{aligned}
  \dot{L} &= \alpha_L\,L\,C \;-\; \mu_L\,L \;+\; \gamma_{RL}\,R\,C, \\
  \dot{R} &= \alpha_R\,R\,C \;-\; \mu_R\,R \;+\; \gamma_{LR}\,L\,C,
\end{aligned}
\end{equation}
where $C=1-L-R$ and all parameters are strictly positive. By construction,
$\dot{C}=-\dot{L}-\dot{R}$, so $L+R+C=1$ is conserved for all $t\ge0$.

The parameters have a direct political interpretation. The terms
$\alpha_L L C$ and $\alpha_R R C$ describe within-bloc recruitment from the
mainstream: existing radical support increases the rate at which mainstream
voters move into the same radical bloc. The terms $-\mu_i X_i$ describe
reabsorption from radical support back into the mainstream. The terms
$\gamma_{ij} X_j C$ describe \emph{reactive polarisation}: the visibility or
support level of one radical wing increases recruitment from the mainstream
into the opposite wing. These parameters are summarised in
Table~\ref{tab:params}.

\begin{table}[ht]
\caption{Parameters of the baseline voter-flow model.}\label{tab:params}
\centering
\renewcommand{\arraystretch}{1.4}
\begin{tabular}{lll}
\toprule
Parameter & Interpretation & Range \\
\midrule
$\alpha_L,\alpha_R$ & Within-bloc recruitment from mainstream & $(0,\infty)$ \\
$\mu_L,\mu_R$       & Reabsorption to mainstream & $(0,\infty)$ \\
$\gamma_{RL}$       & Reactive polarisation: $R\to L$ recruitment & $(0,\infty)$ \\
$\gamma_{LR}$       & Reactive polarisation: $L\to R$ recruitment & $(0,\infty)$ \\
\bottomrule
\end{tabular}
\end{table}

The strict positivity of the cross-couplings $\gamma_{RL}$ and $\gamma_{LR}$
is a structural assumption for the uniqueness results below. If one or
both cross-couplings vanish, the boundary and symmetric limiting cases
can admit non-unique left--right splits;
Appendix~\ref{app:proofs} records the no-cross-coupling limit.

The recruitment terms have the bilinear form of a minimal mass-action closure
compatible with conservation and saturation: recruitment vanishes when either
the recruiting radical bloc or the recruitable mainstream pool is absent. The
reabsorption terms $-\mu_i X_i$ are linear because return to the mainstream
is modelled as an individual-level exit rate from radical support. We use
$\gamma_{RL}RC$ rather than a direct wing-to-wing transfer term such as $RL$:
the intended mechanism is not direct switching between radical camps, but
movement from the mainstream toward one wing in response to the visibility
of the other. Trajectories that start in $\Simp$ remain in $\Simp$ for all
$t\ge0$; the proof is given in Appendix~\ref{app:proofs}.

This baseline model is intentionally minimal. Its role is to establish what a
fixed-parameter voter-flow system can and cannot produce. The next section
shows that, under fixed structural conditions, the baseline dynamics are
path-independent in the long run: temporary state perturbations may alter
the transient trajectory, but they do not change the attracting equilibrium
selected by the parameter vector.

\section{When Does the Mainstream Recover?}\label{sec:baseline}

The substantive question for the baseline model is whether a temporary rise
in radical support changes the long-run political regime or merely displaces
the electorate within a self-correcting system. The answer is governed by a
single threshold balancing radical recruitment against reabsorption into the
mainstream.

Define the recruitment matrix and reabsorption matrix:
\[
  K:=\begin{pmatrix}\alpha_L & \gamma_{RL}\\\gamma_{LR} & \alpha_R\end{pmatrix},
  \qquad
  M:=\begin{pmatrix}\mu_L & 0\\ 0 & \mu_R\end{pmatrix},
\]
and let $\mathcal{R}_{\mathrm{base}}:=\lambda_{\mathrm{PF}}(M^{-1}K)$ denote
the Perron--Frobenius root of the effective recruitment matrix. The quantity $\mathcal{R}_{\mathrm{base}}$ is the model's threshold parameter:
it measures whether radical recruitment, amplified by reactive polarisation,
is strong enough relative to reabsorption to sustain a positive long-run
radical share. It is formally analogous to the basic reproduction number in
compartmental models \citep{diekmann1990,vandendriessche2002}; the
mathematical connection is detailed in Appendix~\ref{app:proofs}.

\begin{theorem}[Baseline threshold and long-run recovery]
\label{thm:baseline_threshold}
Let
\[
  \mathcal{R}_{\mathrm{base}}:=\lambda_{\mathrm{PF}}(M^{-1}K).
\]
For the baseline model~\eqref{eq:model}, exactly two parameter regimes are
possible.
\begin{enumerate}
  \item If $\mathcal{R}_{\mathrm{base}}\le 1$, the radical-free equilibrium
    $E_0=(0,0)$ is globally asymptotically stable on $\Simp$: radical support
    dies out and the entire electorate returns to the mainstream.
  \item If $\mathcal{R}_{\mathrm{base}}>1$, the radical-free equilibrium
    is unstable and there exists a unique interior equilibrium
    $E_1=(L^*,R^*)$. This equilibrium is globally attracting on
    $\Simp\setminus\{E_0\}$, and its mainstream share is
    \[
      C^*=\frac{1}{\mathcal{R}_{\mathrm{base}}}.
    \]
    More explicitly, if $v\gg0$ is the right Perron--Frobenius eigenvector of
    $M^{-1}K$ normalised by $v_L+v_R=1$, then
    \[
      (L^*,R^*)=\left(1-\frac{1}{\mathcal{R}_{\mathrm{base}}}\right)v.
    \]
\end{enumerate}
\end{theorem}

\begin{proof}
At any interior equilibrium, writing $x=(L^*,R^*)^\top$ and
$C^*=1-L^*-R^*$, the equilibrium conditions read $C^*\,Kx=Mx$, or
equivalently $M^{-1}Kx=(C^*)^{-1}x$. Since $x\gg0$, the Perron--Frobenius
theorem implies $(C^*)^{-1}=\lambda_{\mathrm{PF}}(M^{-1}K)$, which gives
$C^*=1/\mathcal{R}_{\mathrm{base}}$ and the formula for $E_1$ in terms of
the right Perron eigenvector. The global stability statements are proved in
Appendix~\ref{app:proofs}: the symmetric case is Theorem~\ref{thm:global}
and the asymmetric case is Theorem~\ref{thm:global_asym}.
\end{proof}

In the symmetric case $\alpha_L=\alpha_R=\alpha$, $\mu_L=\mu_R=\mu$,
$\gamma_{LR}=\gamma_{RL}=\gamma$, set $\beta:=\alpha+\gamma$. The total
radical share $S:=L+R$ then satisfies the scalar logistic-type equation
\[
  \dot S = S\bigl[\beta(1-S)-\mu\bigr],
\]
so $\mathcal{R}_{\mathrm{base}}=\beta/\mu$, and when
$\mathcal{R}_{\mathrm{base}}>1$ the long-run mainstream share at the
interior equilibrium is
\[
  C^* = \frac{\mu}{\beta} = \frac{\mu}{\alpha+\gamma}.
\]
Stronger recruitment ($\alpha$) or stronger reactive polarisation ($\gamma$)
lowers $C^*$; faster reabsorption ($\mu$) raises it. The scalar form of
$\dot S$ shows that the total radical share converges monotonically. In the
full two-dimensional model, trajectories converge to a unique long-run
equilibrium and no periodic oscillations are possible; the absence of
periodic orbits follows from a Bendixson--Dulac argument given in
Appendix~\ref{app:proofs}. (Section~\ref{sec:legacy} specialises to the symmetric parameter regime
with $\delta_L=\delta_R=\delta$; under that regime the total radical share
$S$ satisfies a closed two-dimensional system, and on the equal-wing
subspace $L=R=:P$ one has $S=2P$.)

The substantive implication of Theorem~\ref{thm:baseline_threshold} is a
negative result: under fixed parameters, the baseline model has a unique
attracting equilibrium in each regime. It therefore cannot generate
attracting bistability, path-dependent coexistence of multiple long-run
regimes, or staircase dynamics in radical support under unchanged
parameters. A crisis that temporarily shifts vote shares may produce a
large transient displacement, but once the shock ends the trajectory
returns to the same attracting equilibrium selected by the fixed parameter
vector. The only exceptional case is the nongeneric seed-free state $E_0$
in the supercritical regime: a perturbation that creates the first nonzero
radical seed selects the already-existing attracting equilibrium $E_1$.
This is seeding of an unstable zero state, not a path-dependent change in
the attractor created by the shock. Formally:

\begin{corollary}[No endogenous staircase dynamics under fixed parameters]
\label{cor:no_staircase}
For fixed positive parameters, the baseline model has no attracting
bistability, and no sequence of temporary state perturbations can generate a
sequence of new attracting long-run mainstream shares. Apart from the
nongeneric seed-free equilibrium $E_0$ in the supercritical regime, the
long-run attracting equilibrium is determined by the parameter regime, not
by shock history.
\end{corollary}

\begin{proof}
Immediate from Theorem~\ref{thm:baseline_threshold}: for each fixed positive
parameter set, the baseline dynamics admit a single attracting equilibrium.
\end{proof}

This corollary is the precise mathematical statement of the baseline model's
structural limitation. Explaining the ratchet patterns observed in post-2008
European elections requires a change in model structure. The next section
introduces that extension.

Additional results on local linearisation, the transcritical bifurcation
structure at $\mathcal{R}_{\mathrm{base}}=1$, comparative statics, and the
absence of periodic orbits are collected in Appendix~\ref{app:proofs}.

\section{Extending the Model: Disengagement and Structural
Shocks}\label{sec:extended}

Corollary~\ref{cor:no_staircase} identifies the central limitation of the
baseline model. Under fixed parameters, temporary shocks can change vote
shares, but they cannot generate history-dependent long-run floors,
bistability, or cumulative ratchet effects. To represent crisis episodes, we
add two distinct mechanisms with different roles:
\begin{enumerate}
  \item[(i)] a transient amplifier for crisis-induced disengagement,
    modelled by a compartment $A(t)$ that absorbs part of the mainstream
    during a crisis and is later drained by re-engagement and radical
    mobilisation; this modifies transient dynamics but not the equilibrium
    set;
  \item[(ii)] a permanent shift in the recruitment environment, modelled by
    a structural change in the parameter vector at the time of the crisis;
    this is the ingredient that can change the long-run attractor.
\end{enumerate}

The extended model tracks four groups: a left-radical bloc $L(t)$, a
right-radical bloc $R(t)$, a mainstream bloc $C(t)$, and disengaged voters
$A(t)$. Here $A(t)$ is interpreted as crisis-induced disengagement above
ordinary background abstention. The conservation law becomes
\begin{equation}\label{eq:conserve4}
  L(t)+R(t)+C(t)+A(t)=1 \qquad (t\ge0),
\end{equation}
so that $C=1-L-R-A$ and the feasible state space is
\[
  \mathcal{T}=\{(L,R,A)\in\R^3:\ L\ge0,\;R\ge0,\;A\ge0,\;L+R+A\le1\}.
\]

The dynamics are
\begin{equation}\label{eq:4group}
\begin{aligned}
  \dot{L} &= \alpha_L\,L\,C + \gamma_{RL}\,R\,C + \delta_L\,A\,L - \mu_L\,L, \\
  \dot{R} &= \alpha_R\,R\,C + \gamma_{LR}\,L\,C + \delta_R\,A\,R - \mu_R\,R, \\
  \dot{A} &= \sigma(t)\,C - \delta_L\,A\,L - \delta_R\,A\,R - \rho\,A,
\end{aligned}
\end{equation}
where $C=1-L-R-A$ and all constant parameters are positive. Throughout
the non-impulsive formulation, the crisis input
$\sigma\colon[0,\infty)\to[0,\infty)$ is assumed locally bounded and
piecewise continuous. Equivalently, one may allow locally integrable
nonnegative inputs and interpret solutions in the Carath\'eodory sense.
The extended model is still formulated on a conserved electorate. What changes is not the size of
the electorate but its internal composition: crises can now shift voters into
disengagement, and some of those voters may later be reabsorbed into the
mainstream or mobilised by radical blocs. When $\sigma(t)\equiv0$, the face $\{A=0\}$ is forward-invariant; any
trajectory starting with $A(0)=0$ remains there and the model reduces
exactly to the baseline system~\eqref{eq:model}. Moreover, every equilibrium
of the autonomous post-shock system has $A^*=0$, since
\[
  \dot A=-A(\delta_LL+\delta_RR+\rho)
\]
and the bracket is strictly positive (as $\rho>0$). Hence, for fixed
post-shock parameters, the equilibrium set of the extended system coincides
with the equilibrium set of the corresponding baseline system on the face
$A=0$: the disengagement compartment can amplify transients, but it does
not by itself create a new long-run equilibrium once the crisis input has
ended. Trajectories that start in $\mathcal{T}$ remain in $\mathcal{T}$ for
all $t\ge0$; the proof is given in Appendix~\ref{app:proofs4}.

The new terms have a direct political interpretation. The bilinear factors
$\delta_L A L$ and $\delta_R A R$ describe mobilisation of disengaged voters
into existing radical blocs. As in the baseline recruitment terms, mobilisation
requires both a recruitable pool and a pre-existing radical seed.
Consequently, if $L(t_0)=R(t_0)=0$ then $L(t)\equiv R(t)\equiv 0$ for all
later times under the dynamics considered here, regardless of the size of
$A(t_0)$ and regardless of subsequent parameter shifts within this model
class. A radical movement must therefore have some pre-existing
organisational or electoral core before it can mobilise disengaged voters;
the model does not describe how such a seed forms. The term $\rho A$
describes re-engagement from disengagement into the mainstream, and
$\sigma(t)C$ moves part of the current mainstream pool into disengagement
during a crisis. The additional parameters are summarised in
Table~\ref{tab:params4}.

\begin{table}[ht]
\caption{New parameters in the extended model. Surge conditions on
$\delta_L,\delta_R$ are imposed in Section~\ref{sec:legacy}.}\label{tab:params4}
\centering
\renewcommand{\arraystretch}{1.4}
\begin{tabular}{lll}
\toprule
Parameter & Interpretation & Range \\
\midrule
$\delta_L,\delta_R$ & Mobilisation: disengagement $\to$ radical & $(0,\infty)$ \\
$\rho$ & Re-engagement: disengagement $\to$ mainstream & $(0,\infty)$ \\
$\sigma(t)$ & Crisis input: mainstream $\to$ disengagement & $\sigma(t)\ge0$ \\
\bottomrule
\end{tabular}
\end{table}

The global equilibrium classification below does not require any ordering
between $\delta_i$ and the baseline recruitment parameters. Additional
inequalities are needed only for the transient-surge calculation. In the
symmetric reduction, the relevant comparison is not $\delta>\alpha$ but
$\delta>\beta=\alpha+\gamma$: the disengaged pool must be a more effective
source of radical mobilisation than the mainstream pool. A feasible
near-mainstream state-shock threshold $\Delta_c\in(0,1)$ requires the
stronger regime $\delta>\mu>\beta$. In the asymmetric model, the analogous
condition is the spectral inequality $\Phi(1)>1$
(Theorem~\ref{thm:PF_Delta} in Appendix~\ref{app:proofs4}), which does not
in general reduce to a componentwise comparison between $\delta_i$ and
$\alpha_i+\gamma_{ji}$.

Several modelling choices delimit the scope of the extended system.
Crisis-induced disengagement is drawn only from $C$ in keeping with the
definition of $A(t)$ as excess disengagement from mainstream party
competition; the asymmetry between mainstream and committed radical
attachment is consistent with survey evidence on differential stability of
party identification across blocs \citep{dalton2004,gles2021}. There is no
reactive mobilisation term of the form $\gamma^A_{ij}X_jA$: mobilisation
from $A$ is modelled as intra-bloc mobilisation by existing radical support,
while cross-wing reactive polarisation is routed through the mainstream
pool as in the baseline system. Finally, parameters are constant within
each post-shock regime; structural crises are represented as discrete jumps
in the parameter vector (Definition~\ref{def:structural}) rather than as
continuous parameter drift. These restrictions isolate the minimal
mechanism needed to separate transient amplification from long-run
realignment.

A temporary crisis shock is represented as an instantaneous transfer of part
of the current mainstream pool into disengagement.

\begin{definition}[State shock]\label{def:impulse}
At time $t_0$, a \emph{state shock} of amplitude $\Delta\in(0,1)$ is the jump
\begin{equation}\label{eq:jump}
  C(t_0^+)=(1-\Delta)\,C(t_0^-),\qquad
  A(t_0^+)=A(t_0^-)+\Delta\,C(t_0^-),
\end{equation}
with $L$ and $R$ unchanged. After the jump the system evolves autonomously
with $\sigma\equiv0$.
\end{definition}

\begin{remark}[Dirac representation of the impulse]
Equivalently, the jump can be represented by a Dirac impulse in $\sigma$ of
mass $s$ at $t_0$, with $\sigma$ otherwise zero, related to the amplitude
by $\Delta=1-e^{-s}$. We use the jump form throughout.
\end{remark}

\begin{definition}[Structural shock]\label{def:structural}
Let $\theta$ denote the vector of constant model parameters. A
\emph{structural shock} at time $t_0$ is a permanent change
\[
  \theta^-\longrightarrow\theta^+
\]
that remains in force for all $t>t_0$. It may, but need not, be accompanied
by the state shock in Definition~\ref{def:impulse}. We distinguish three
cases:
\begin{itemize}
  \item if $\theta^+=\theta^-$, the event is a \emph{pure state shock};
  \item if the state is continuous at $t_0$ but $\theta^+\ne\theta^-$, the
    event is a \emph{pure structural shock};
  \item if both components are present, the event is a \emph{mixed shock}.
\end{itemize}
In the symmetric reduction considered below, the structural component is
represented by a scalar shift $\beta^+=\beta^-+\Delta\beta$, with the
remaining parameters unchanged.
\end{definition}

Substantively, the distinction is between crises that temporarily redistribute
voters and crises that durably alter the environment in which radical
recruitment takes place. The analytical consequences of this distinction are
the subject of the next section.

\section{When Do Crises Leave a Lasting Political Legacy?}\label{sec:legacy}

The previous section introduced the extended model and the distinction
between state shocks and structural shocks. This section derives the
consequences of that distinction. Three results are central. First, after
the crisis input has ended, the autonomous system has the same threshold
structure as the baseline model: disengagement disappears at equilibrium
and the long-run regime is governed by a Perron--Frobenius recruitment
threshold. Second, in the symmetric reduction, a pure state shock can
generate a transient radical surge; the surge has an explicit critical
amplitude and a finite mobilisation-window bound. Third, a lasting change
in the long-run radical-support floor occurs only when the structural
component of the shock moves the post-crisis parameters across the
threshold. Thus disengagement governs transient amplification, whereas
structural parameter shifts govern long-run realignment.

Throughout this section, unless otherwise stated, all parameters are the
post-shock parameters of the autonomous system after the crisis input has
ended ($\sigma\equiv0$). For the post-shock system \eqref{eq:4group}, set
\[
  K:=\begin{pmatrix}
    \alpha_L & \gamma_{RL}\\
    \gamma_{LR} & \alpha_R
  \end{pmatrix},
  \qquad
  M:=\begin{pmatrix}
    \mu_L & 0\\
    0 & \mu_R
  \end{pmatrix},
\]
and define the radical-support threshold
\[
  \mathcal R_{\mathrm{rad}}:=\lambda_{\mathrm{PF}}(M^{-1}K).
\]
This is the extended model's analogue of the baseline threshold
$\mathcal R_{\mathrm{base}}$: it compares recruitment and reactive
polarisation against reabsorption into the mainstream. Note that
$\mathcal R_{\mathrm{rad}}$ depends on recruitment, polarisation, and
reabsorption parameters, but not on the mobilisation rates $\delta_i$ or
the re-engagement rate $\rho$. The latter affect the transient use and
depletion of the disengagement reservoir; they do not affect the
equilibrium threshold once the crisis input has ended.

\begin{theorem}[Global post-shock dynamics]\label{thm:global4}
Consider the autonomous post-shock system \eqref{eq:4group} with
$\sigma\equiv0$.
\begin{enumerate}
  \item If $\mathcal R_{\mathrm{rad}}\le1$, the radical-free equilibrium
  $E_0=(L,R,A)=(0,0,0)$ is globally asymptotically stable on $\mathcal T$.
  \item If $\mathcal R_{\mathrm{rad}}>1$, there is a unique positive
  radical-support equilibrium $E_1$. Explicitly, if $v\gg0$ is the right
  Perron--Frobenius eigenvector of $M^{-1}K$ normalised by $v_L+v_R=1$, then
  \[
    (L^*,R^*)=\left(1-\frac{1}{\mathcal R_{\mathrm{rad}}}\right)v,
    \qquad
    C^*=\frac{1}{\mathcal R_{\mathrm{rad}}},
    \qquad
    A^*=0.
  \]
  The radical-free axis
  $\Gamma:=\{(L,R,A)\in\mathcal T:\ L=R=0\}$
  is forward invariant and converges to $E_0$. Every trajectory starting in
  $\mathcal T\setminus\Gamma$ converges to $E_1$.
\end{enumerate}
\end{theorem}

\begin{proof}
See Appendix~\ref{app:proofs4}.
\end{proof}

Theorem~\ref{thm:global4} shows that the disengaged compartment is transient
at equilibrium. It can affect the path by which the system moves, but the
long-run regime is still determined by the structural recruitment threshold.
In the asymmetric model this threshold is Perron--Frobenius; in the symmetric
case it reduces to a scalar condition.

For closed-form shock calculations we use the symmetric parameter regime
\[
  \alpha_L=\alpha_R=\alpha,\qquad
  \mu_L=\mu_R=\mu,\qquad
  \gamma_{LR}=\gamma_{RL}=\gamma,\qquad
  \delta_L=\delta_R=\delta,
\]
and write
\[
  S(t):=L(t)+R(t),\qquad \beta:=\alpha+\gamma.
\]
A direct computation shows that the total radical share $S$ and the
disengagement share $A$ then satisfy the closed two-dimensional system
\begin{equation}\label{eq:2d}
\begin{aligned}
  \dot{S} &= S\bigl[(\beta-\mu)-\beta S+(\delta-\beta)A\bigr],\\
  \dot{A} &= -(\delta S+\rho)A,
\end{aligned}
\end{equation}
on the feasible region
$\mathcal F=\{(S,A):\ S\ge0,\ A\ge0,\ S+A\le1\}$. This reduction does not
require an equal-wing initial condition: under symmetric parameters, the
$(S,A)$-system is closed for arbitrary initial $(L_0,R_0,A_0)$, and the
wing imbalance $L-R$ decouples (see Appendix~\ref{app:proofs4}). On the
invariant equal-wing subspace $L=R=:P$ one has $S=2P$, recovering the
per-wing formulation.

\begin{theorem}[Equilibria of the symmetric extended model]\label{thm:eq4}
System \eqref{eq:2d} has no equilibrium with $A^*>0$. Its equilibria are
$(S,A)=(0,0)$ and, when $\beta>\mu$,
\[
  (S^*,A^*)=\left(1-\frac{\mu}{\beta},\,0\right).
\]
The corresponding equilibrium of the full symmetric four-group system on
the equal-wing subspace $L=R$ is
\[
  L^*=R^*=\tfrac12\!\left(1-\tfrac{\mu}{\beta}\right),
  \qquad
  C^*=\tfrac{\mu}{\beta},
  \qquad
  A^*=0.
\]
\end{theorem}

\begin{proof}
Since $\dot A=-(\delta S+\rho)A$ and $\delta S+\rho>0$, every equilibrium
has $A^*=0$. Substitution into $\dot S=0$ gives $S^*=0$ or
$S^*=1-\mu/\beta$, the latter feasible iff $\beta>\mu$.
\end{proof}

\begin{remark}[Interpretation of $A(t)\to0$]\label{rem:Abase}
The variable $A(t)$ represents excess crisis-induced disengagement above
ordinary background abstention. Thus $A(t)\to0$ means that the post-crisis
disengagement reservoir has been depleted: voters have either been reabsorbed
into the mainstream or mobilised into radical support. The transient phase can
nevertheless be politically important, especially when the re-engagement rate
$\rho$ is small.
\end{remark}

The stability classification mirrors the baseline model. If $\beta\le\mu$,
the origin $(S,A)=(0,0)$ is globally asymptotically stable on $\mathcal F$.
If $\beta>\mu$, the edge $\{S=0\}$ is invariant; every trajectory with
$S_0>0$ converges to the positive equilibrium of Theorem~\ref{thm:eq4}.
The proof is given in Appendix~\ref{app:proofs4}.

By the bilinear-mobilisation observation of Section~\ref{sec:extended}, the
shock results below are conditional on a nonzero radical seed,
$S_0=L_0+R_0>0$. The radical-free axis $\{S=0\}$ is invariant: a state
shock into a zero radical baseline cannot generate radical support,
regardless of the size of $A$ or any subsequent parameter shift.

We next quantify when a state shock creates a transient radical surge.
Suppose the system is initially close to the radical-free state, with
$S(t_0^-)=\varepsilon>0$, $A(t_0^-)=0$, and $0<\varepsilon\ll1$. A state
shock of amplitude $\Delta$ places $A(t_0^+)=\Delta\,C(t_0^-)$.

\begin{theorem}[Critical state-shock amplitude]\label{thm:critical}
Assume $\beta<\mu$. In the near-mainstream limit $\varepsilon\to0$, the
total radical share initially grows after the state shock if and only if
\begin{equation}\label{eq:growth_condition}
  (\beta-\mu)+(\delta-\beta)\Delta>0.
\end{equation}
Equivalently:
\begin{enumerate}
  \item if $\delta\le\beta$, no feasible state shock produces initial
  near-mainstream radical growth;
  \item if $\delta>\beta$, initial growth occurs if and only if
  \begin{equation}\label{eq:Delta_c}
    \Delta>\Delta_c:=\frac{\mu-\beta}{\delta-\beta};
  \end{equation}
  \item the threshold is feasible, $\Delta_c\in(0,1)$, if and only if
  $\delta>\mu>\beta$.
\end{enumerate}
\end{theorem}

\begin{proof}
From \eqref{eq:2d}, $\dot S/S=(\beta-\mu)-\beta S+(\delta-\beta)A$. At the
post-shock state and in the limit $\varepsilon\to0$ this equals
$(\beta-\mu)+(\delta-\beta)\Delta$, giving \eqref{eq:growth_condition}.
For (i), the left-hand side is nonpositive for all $\Delta\in[0,1]$ when
$\delta\le\beta$. For (ii), rearrangement of \eqref{eq:growth_condition}
under $\delta>\beta$ yields $\Delta>(\mu-\beta)/(\delta-\beta)=\Delta_c$.
For (iii), $\Delta_c<1$ iff $\mu-\beta<\delta-\beta$, i.e.\ $\delta>\mu$.
\end{proof}

The politically relevant regime $\delta>\mu>\beta$ admits a direct
interpretation. Stronger reabsorption $\mu$ raises the resilience
threshold; stronger mobilisation $\delta$ lowers it; proximity to the
structural threshold $\beta=\mu$ makes the system increasingly sensitive,
$\Delta_c\to 0^+$ as $\beta\to\mu^-$.

\begin{theorem}[Mobilisation-window bound]\label{thm:window}
Assume $\beta<\mu<\delta$, and let $\Delta_c=(\mu-\beta)/(\delta-\beta)$.
For a pure state shock of amplitude $\Delta>\Delta_c$ at time $t_0$ with
$A(t_0^-)=0$, set $\tau:=t-t_0$ and
\begin{equation}\label{eq:tstar}
  \tau^*:=\frac{1}{\rho}\ln\!\left(\frac{\Delta}{\Delta_c}\right)
        =\frac{1}{\rho}\ln\!\left(\frac{(\delta-\beta)\Delta}{\mu-\beta}\right).
\end{equation}
Then $\dot S(\tau)<0$ for every $\tau>\tau^*$ with $S(\tau)>0$. In
particular, the radical surge window is contained in $[0,\tau^*]$.
\end{theorem}

\begin{proof}
Since $\dot A=-(\delta S+\rho)A\le-\rho A$, integration gives
$A(\tau)\le\Delta e^{-\rho\tau}$. Using \eqref{eq:2d},
\[
  \frac{\dot S(\tau)}{S(\tau)}
  =(\beta-\mu)-\beta S(\tau)+(\delta-\beta)A(\tau)
  \le(\beta-\mu)+(\delta-\beta)\Delta e^{-\rho\tau},
\]
because $-\beta S\le0$. The right-hand side equals zero at $\tau=\tau^*$
and is strictly negative for $\tau>\tau^*$.
\end{proof}

The bound holds without a small-$S$ approximation: the term $-\beta S$
strictly accelerates the suppression of radical growth.

\begin{corollary}[Comparative statics of the mobilisation window]\label{cor:levers}
The window bound $\tau^*$ in Theorem~\ref{thm:window} depends on three
quantities: the re-engagement rate $\rho$, which enters as $1/\rho$; the
shock ratio $\Delta/\Delta_c$, which enters logarithmically; and the
structural distance from threshold $\mu-\beta$, which becomes decisive
when $\beta$ is close to $\mu$. Increasing re-engagement shortens the
window directly, while reducing baseline recruitment and reactive
polarisation both raises $\Delta_c$ and shortens the transient phase.
\end{corollary}

The preceding results concern transient growth. They do not determine whether
the crisis leaves a permanent political legacy. That question is governed by
the structural component of the shock.

\begin{theorem}[Bifurcation trigger]\label{thm:trigger}
Consider the symmetric system with fixed $\mu$ and initial $\beta_0<\mu$,
in which the structural component of the shock is restricted to a scalar
shift $\beta_0\to\beta_0+\Delta\beta$, with $\mu,\delta,\rho$ unchanged.
For any nonzero radical seed $S_0>0$, the structural shock produces a
permanent shift in the long-run regime if and only if
\begin{equation}\label{eq:trigger}
  \beta_0+\Delta\beta>\mu.
\end{equation}
When \eqref{eq:trigger} holds, the post-shock long-run mainstream share is
\begin{equation}\label{eq:Cinf}
  C^\infty=\frac{\mu}{\beta_0+\Delta\beta}<1,
\end{equation}
and this value is independent of the state-shock amplitude $\Delta$. If
the radical seed is exactly zero, the radical-free axis $\{S=0\}$ remains
invariant and the system returns to the radical-free state regardless
of $\Delta\beta$.
\end{theorem}

\begin{proof}
See Appendix~\ref{app:proofs4}.
\end{proof}

\begin{remark}[Pure structural shock]\label{rem:pure_structural}
Theorem~\ref{thm:trigger} applies in particular to pure structural shocks
($\Delta=0$ in Definition~\ref{def:structural}). In that case the state is
continuous at $t_0$ and the bifurcation is triggered by the parameter
shift alone, without any turnout disruption. The structural component
$\Delta\beta$ is therefore both necessary and sufficient for a permanent
shift in the long-run regime; the state component governs only the transient
phase.
\end{remark}

\begin{remark}[The fundamental decomposition]\label{rem:decomposition}
Theorem~\ref{thm:trigger} separates the two components of a crisis. The state
component $\Delta$ determines whether radical support rises temporarily and
how long the mobilisation window remains open. It does not determine the
long-run equilibrium. The structural component $\Delta\beta$ determines
whether the long-run equilibrium itself changes. A society can therefore
manage the visible crisis shock while still suffering a durable political
realignment if the crisis permanently shifts the recruitment environment.
\end{remark}

For any nonzero radical seed, once $\beta_0+\Delta\beta>\mu$, return to the
pre-shock radical-free equilibrium is impossible under the autonomous
post-shock dynamics. Recovery requires a structural counter-shock: a
sustained reduction in recruitment or reactive polarisation, or an
increase in reabsorption, that moves the system back below threshold. The
disengaged compartment is not itself the source of irreversibility: it
amplifies transient surges, but lasting change comes from permanent shifts
in the structural parameters.

The asymmetric version of this statement replaces the scalar condition
$\beta>\mu$ by the Perron--Frobenius condition
$\lambda_{\mathrm{PF}}(M_{\mathrm{post}}^{-1}K_{\mathrm{post}})>1$.
The full asymmetric threshold and the corresponding asymmetric shock-window
bounds are stated in Appendix~\ref{app:proofs4}.

Finally, consider a sequence of crises. Let shocks occur at times
$t_1<t_2<\cdots<t_n$, with state components $\Delta_1,\ldots,\Delta_n$ and
nonnegative structural components $\Delta\beta_1,\ldots,\Delta\beta_n$.
Define the cumulative structural parameter after $k$ shocks:
$B_k:=\beta_0+\sum_{j=1}^k\Delta\beta_j$.

\begin{theorem}[Cumulative structural shocks]\label{thm:cumulative}
Assume $\beta_0<\mu$, $\Delta\beta_k\ge0$ for all $k$, and that the
radical seed is nonzero at the first threshold-crossing shock
(equivalently, $S(t_{k^*}^+)>0$ for the shock at which the supercritical
regime is first entered).
\begin{enumerate}
  \item If $B_k\le\mu$, the long-run mainstream share after the $k$-th
  shock is $C_k^\infty=1$. State shocks may create transient surges but do
  not alter the long-run regime.
  \item If $k^*$ is the first index such that $B_{k^*}>\mu$, the $k^*$-th
  shock crosses the structural threshold and the long-run mainstream share
  becomes $C_{k^*}^\infty=\mu/B_{k^*}<1$.
  \item For every later shock with $B_k>\mu$,
  $C_k^\infty=\mu/B_k\le C_{k-1}^\infty$, with strict inequality exactly
  when $\Delta\beta_k>0$.
  \item The threshold-crossing index is
  \begin{equation}\label{eq:kstar}
    k^*=\min\!\left\{k:\sum_{j=1}^k\Delta\beta_j>\mu-\beta_0\right\}.
  \end{equation}
\end{enumerate}
\end{theorem}

\begin{proof}
See Appendix~\ref{app:proofs4}.
\end{proof}

The monotone staircase conclusion uses the assumption $\Delta\beta_k\ge0$.
Allowing negative structural components would represent structural
counter-shocks and could reverse some steps or restore the subcritical
regime.

When $B_k<\mu$ and $\delta>B_k$, the transient near-mainstream surge
threshold after the $k$-th shock is
$\Delta_c^{(k)}=(\mu-B_k)/(\delta-B_k)$,
evaluated under the post-shock structural parameter $B_k$. A crisis can
thus create a visible radical surge while leaving no permanent legacy if
$B_k\le\mu$.

\begin{corollary}[Staircase dynamics]\label{cor:staircase}
Under the conditions of Theorem~\ref{thm:cumulative}, after the first
threshold-crossing shock $k^*$, every subsequent shock with $\Delta\beta_k>0$
strictly lowers the long-run mainstream share: $C_k^\infty<C_{k-1}^\infty$.
If all subsequent structural components are positive, the sequence
$C_{k^*}^\infty>C_{k^*+1}^\infty>\cdots$ is strictly decreasing. This is
the staircase pattern: repeated crises leave radical support at successively
higher long-run floors.
\end{corollary}

\begin{proof}
Immediate from $C_k^\infty=\mu/B_k$ and $B_k=B_{k-1}+\Delta\beta_k$.
\end{proof}

The next section presents a single descriptive case---German federal
elections, 2013--2025---as a stylised illustration.

\section{Stylised Illustration: German Federal Elections, 2013--2025}\label{sec:germany}

To illustrate how the model's distinction between transient and
structural mechanisms can organise an observed electoral pattern, we
apply the framework to a single descriptive case: German federal
elections, 2013--2025. The aim is illustrative; no causal estimation is
attempted, and no claim is made about the full set of mechanisms driving
radical-party support in Germany.

Within this descriptive frame, the Bundestag second-vote share for the
right-radical bloc (Alternative für Deutschland) followed a pattern
consistent with the staircase mechanism of Section~\ref{sec:legacy}: it
rose from 4.7\% in 2013 to 12.6\% in 2017, fell partially to 10.3\% in
the original 2021 final result, and reached 20.8\% in 2025
\citep{bundeswahlleiter2013,bundeswahlleiter2017,%
bundeswahlleiter2021,bundeswahlleiter2025}. The 2021 dip is not an
anomaly: Corollary~\ref{cor:staircase} governs long-run floors rather
than monotone election-by-election observations.

This section offers a stylised illustration rather than a formal estimation
exercise. The goal is to show how the model's distinction between temporary
state shocks and permanent structural shifts maps onto this sequence. No
parameters are estimated and no econometric identification is claimed.

The model is formulated on a conserved electorate, so the empirical
mapping uses the eligible electorate as the common reference population.
The 2021 figures use the original 2021 final result (turnout 76.6\%, AfD
10.3\%), before the 2024 repeat election in parts of Berlin, in order to
keep the four-election sequence internally consistent. Using the
post-repeat 2024 adjustment changes entries only at the third decimal
place and leaves all qualitative comparisons unchanged.
Let $N(t)=1-\text{turnout}$ denote the observed non-voting share. Using
the rounded official second-vote shares, we define
\begin{align}
  R(t) &= \bigl(\text{AfD second-vote share}\bigr)\times\text{turnout},
  \notag\\
  L(t) &= \bigl(\text{Die Linke second-vote share}\bigr)\times\text{turnout},
  \notag\\
  N(t) &= 1-\text{turnout},
  \notag\\
  C_{\mathrm{obs}}(t) &= 1-L(t)-R(t)-N(t).
  \label{eq:proxy}
\end{align}
Thus $V(t)=L(t)+R(t)$ is an approximate observed radical vote share as a
fraction of the eligible electorate, and the residual category
$C_{\mathrm{obs}}$ collects all remaining parties, invalid second votes,
and rounding residuals.

At the measurement level, these proxies relate to the latent ODE states
through a measurement convention. Reading $L_t,R_t,A_t$ as the latent
ODE states (rather than the rounded proxies above), a compact convention
is
\begin{equation}\label{eq:measurement}
  V_t^{\mathrm{obs}}=L_t+R_t+\varepsilon^V_t,\qquad
  N_t^{\mathrm{obs}}=N^{\mathrm{bg}}+\kappa A_t+\varepsilon^N_t,
  \qquad
  C_t^{\mathrm{obs}}=1-V_t^{\mathrm{obs}}-N_t^{\mathrm{obs}},
\end{equation}
where $N^{\mathrm{bg}}$ is ordinary background non-voting, $\kappa>0$ is
a measurement scale, and the $\varepsilon$ terms collect aggregation and
classification error. The descriptive table below sets these errors to
zero and uses rounded official quantities; no claim is made that $A_t$,
$N^{\mathrm{bg}}$, or $\kappa$ are identified from four elections. The left-radical proxy is Die Linke throughout 2013--2025. The 2025 vote
for Bündnis Sahra Wagenknecht (BSW) is kept in the residual category in
the baseline mapping. This avoids introducing a new party into the
left-radical series only in the final observation and does not require
taking a position on BSW's classification. The convention is conservative
for the comparison made here: including BSW would raise $V(2025)$ and
therefore strengthen, rather than weaken, the qualitative staircase
pattern. We therefore treat the BSW-inclusive mapping as a robustness
check rather than as the baseline series.

The model variable $A(t)$ corresponds not to total non-voting $N(t)$, but
to the excess crisis-induced disengagement component of $N(t)$. We use $N(t)$
as the directly observed quantity and define an illustrative excess component
below by subtracting a background non-voting level.

\begin{table}[ht]
  \centering
\caption{Bundestag elections 2013--2025, as fractions of the eligible
    electorate. $V=L+R$ uses AfD and Die Linke as radical-party proxies
    (BSW 2025 in residual category); $C_{\mathrm{obs}}$ is mainstream and
    residual support; $N$ is non-voting; Sum is $V+C_{\mathrm{obs}}+N$.
    Source: \citet{bundeswahlleiter2013,bundeswahlleiter2017,%
bundeswahlleiter2021,bundeswahlleiter2025}.}
  \label{tab:germany}
  \smallskip
\begin{tabular}{lcccc}
    \hline
    Year & $V=L+R$ & $C_{\mathrm{obs}}$ & $N$ & Sum \\
    \hline
    2013 & 0.095 & 0.620 & 0.285 & 1.000 \\
    2017 & 0.166 & 0.596 & 0.238 & 1.000 \\
    2021 & 0.116 & 0.650 & 0.234 & 1.000 \\
    2025 & 0.244 & 0.581 & 0.175 & 1.000 \\
    \hline
\end{tabular}
\end{table}

\begin{table}[ht]
  \centering
  \caption{Sensitivity of the radical-support proxy $V$ as a fraction of
  the eligible electorate. The baseline broad-radical proxy is
  AfD${}+{}$Die Linke; BSW is added only as a 2025 sensitivity check.}
  \label{tab:germany_robust}
  \smallskip
  \begin{tabular}{lccc}
    \hline
    Year & AfD only & AfD${}+{}$Die Linke & AfD${}+{}$Die Linke${}+{}$BSW \\
    \hline
    2013 & 0.034 & 0.095 & 0.095 \\
    2017 & 0.096 & 0.166 & 0.166 \\
    2021 & 0.079 & 0.116 & 0.116 \\
    2025 & 0.172 & 0.244 & 0.285 \\
    \hline
  \end{tabular}
\end{table}

The robustness table shows that the qualitative 2025 increase is not an
artefact of aggregating AfD with Die Linke. The AfD-only
eligible-electorate series displays the same broad pattern; including
BSW in 2025 raises the broad-radical proxy further. The subsequent
illustrative parameterisation uses the baseline AfD${}+{}$Die Linke
mapping in Table~\ref{tab:germany}.

The model predicts $A^*=0$ at every post-shock equilibrium
(Theorem~\ref{thm:global4}; in the symmetric reduction,
Theorem~\ref{thm:eq4}). The raw observable $N=1-\text{turnout}$ is not the
model variable $A(t)$: it mixes ordinary background non-voting with any
crisis-induced excess disengagement. As a purely descriptive
normalisation, we treat the lowest observed non-voting level in the
series as a background level. With $N_{\mathrm{bg}}\approx0.175$ from
2025, the excess proxy
\[
  A_{\mathrm{exc}}(t):=N(t)-N_{\mathrm{bg}}
\]
is approximately $0.110$ (2013), $0.063$ (2017), $0.059$ (2021), and
$0.000$ (2025). This proxy should not be read as an estimate of the model
state $A(t)$; it is used only to illustrate the distinction between
background abstention and an excess disengagement reservoir.

The decline in observed non-voting from 2021 to 2025 coincides with the
largest single-period increase in $V$. In the language of the model, this
is compatible with the $A\to L\cup R$ mobilisation channel: part of an
excess disengagement reservoir may be reabsorbed into electoral
participation, and some of that mobilisation may flow into radical
support. The aggregate data do not identify these flows; they only show
that the turnout increase and the radical-support increase occurred in
the same interval.

We next give an illustrative parameterisation that reproduces two
post-election reference levels.  Fix $\mu=0.22$ and $\beta_0=0.18$,
corresponding to a subcritical baseline
($\mathcal{R}_{\mathrm{rad}}=0.82<1$; $V^*_0=0$). This parameter choice
should not be interpreted as a fit to the full pre-2013 German party
system. In particular, the observed proxy $V=L+R$ contains a
long-standing left-party component, whereas the parameterisation below
is intended only as a stylised representation of the post-2013
right-radical surge and its subsequent reference levels. The numerical
values therefore illustrate the threshold mechanism; they are not
estimates of the German party system as a whole. After the 2015
migration crisis, set
\begin{equation}
  \beta_1 = 0.249, \qquad
  \mathcal{R}_{\mathrm{rad}} = \tfrac{\beta_1}{\mu} = 1.13 > 1,\qquad
  V^*_1 = 1 - \tfrac{\mu}{\beta_1} = 0.1165.
  \label{eq:beta1}
\end{equation}
After the 2022 energy and cost-of-living crisis, set
\begin{equation}
  \beta_2 = 0.2910, \qquad
  \mathcal{R}_{\mathrm{rad}} = 1.32 > 1, \qquad
  V^*_2 = 1 - \tfrac{\mu}{\beta_2} = 0.2440.
  \label{eq:beta2}
\end{equation}
The two structural shifts are $\Delta\beta_1=\beta_1-\beta_0=0.069$
(2015) and $\Delta\beta_2=\beta_2-\beta_1=0.042$ (2022).

\begin{remark}[Status of the illustrative parameterisation]
  The values of $\beta_1$ and $\beta_2$ in~\eqref{eq:beta1}--\eqref{eq:beta2}
  are reverse-engineered from the two post-election reference levels
  $V(2021)$ and $V(2025)$.  The parameterisation therefore has no
  independent predictive content for these two levels.  Its purpose is only to show that
  the model can represent a subcritical pre-shock regime followed by two
  higher post-shock floors.  The remaining comparisons below should be read
  as qualitative consistency checks, not as statistical tests or forecasts.
\end{remark}

Under the calibration~\eqref{eq:beta1}--\eqref{eq:beta2}, the model
suggests three qualitative consistency checks against Table~\ref{tab:germany}.

\begin{enumerate}[label=(P\arabic*)]

\item \emph{Staircase of long-run floors.}
  Corollary~\ref{cor:staircase} predicts a strictly increasing sequence
  of post-shock equilibria: $V^*_0=0<V^*_1<V^*_2$. The relevant
  descriptive comparison is between later elections after the initial
  crisis episodes: 2021 after the post-2015 shift and 2025 after the
  post-2022 shift. The observed sequence
  $V(2021)=0.116<V(2025)=0.244$ is consistent with the staircase
  prediction. Since $V^*_1$ and $V^*_2$ are reverse-engineered from
  these two observations, this comparison is a consistency check rather
  than an independent test.

\item \emph{Transient overshoot above the new floor.}
  Theorem~\ref{thm:critical} gives the condition under which a state
  shock produces initial growth of total radical support, and
  Theorem~\ref{thm:window} bounds the duration of that growth. For a
  sufficiently large state component, the model thus allows a temporary
  overshoot above the new post-shock floor before convergence back to
  that floor. The 2017 observation $V(2017)=0.166>V^*_1=0.116$ is
  consistent with such a transient overshoot phase; by 2021,
  $V=0.116\approx V^*_1$. The 2021 dip relative to 2017 therefore need
  not be treated as an anomaly within this model class.

\item \emph{Co-movement of $V$ and observed non-voting.}
  The model implies an inverse movement between total radical support
  and the excess disengagement reservoir when disengaged voters are
  reabsorbed into participation and some are mobilised into radical
  support. In the descriptive series, observed non-voting $N$ falls from
  $0.285$ (2013) to approximately $0.175$ (2025), while $V$ rises from
  $0.095$ to approximately $0.244$. The largest single-period
  co-movement occurs between 2021 and 2025, when $N$ falls by $0.059$
  and $V$ rises by $0.128$. This is compatible with the
  $\delta$-mobilisation channel, but the aggregate data do not identify
  the underlying voter flows.

\end{enumerate}

We emphasise that this is a \emph{stylised calibration}: the parameter
values are chosen to reproduce the qualitative regime structure
(subcritical pre-2015, supercritical post-2015, staircase after 2022),
not to minimise any statistical criterion.  A rigorous estimation
procedure would require a probabilistic observation model and is beyond
the scope of the present paper.

\section{Discussion}\label{sec:discussion}

We developed and analysed two coupled ODE models of political competition with
a conserved electorate.

The baseline three-group model is path-independent under fixed parameters:
in both the symmetric and asymmetric cases the dynamics converge globally to
the attractor selected by the Perron--Frobenius threshold. Periodic orbits are
excluded by the Dulac function $B=(LRC)^{-1}$. Consequently, the baseline
cannot generate endogenous staircase dynamics, history-dependent long-run
floors, or multiple attractors.

The four-group extension adds a disengaged compartment and distinguishes pure
state shocks from permanent structural parameter shifts.  The post-shock
dynamics admit a complete global classification governed by the spectral
threshold $\mathcal{R}_{\mathrm{rad}}=\lambda_{\mathrm{PF}}(M^{-1}K)$:
below it the radical-free equilibrium is globally asymptotically stable; above it every
trajectory with a nonzero radical seed converges to the unique radicalised
equilibrium.  Lasting long-run shifts under fixed post-shock parameters arise from
structural parameter changes crossing this threshold; cumulative
sub-threshold shifts produce staircase dynamics absent from the baseline.  The empirical
illustrations against German Bundestag data 2013--2025 are qualitative only
and do not constitute a statistical test.

The model raises several questions of identifiability and estimation that
we leave to future work. A rigorous quantitative calibration would require
a probabilistic
observation model linking the ODE state to electoral and survey data,
and faces the following identifiability obstacles.

In the symmetric three-group model, $\alpha$ and $\gamma$ enter the
total radical-share equation only through their sum $\beta=\alpha+\gamma$.
Equilibrium observations constrain the ratio $\mu/\beta$, since
$V^*=1-\mu/\beta$, but they do not separately identify $\alpha$, $\gamma$,
and $\mu$. Separating direct recruitment from reactive polarisation would
require information on flows between party blocs, or additional survey
measures of cross-wing reactive mobilisation. Panel surveys such as the
German Longitudinal Election Study (GLES) are a natural source for such
auxiliary information~\cite{gles2021}.

The four-group model adds the mobilisation rate $\delta$ and the
re-engagement rate $\rho$. Observed non-voting $N(t)$ can help constrain
these parameters only if it is linked to the latent state $A(t)$ by an
explicit measurement equation, because $N(t)$ includes both ordinary
background abstention and crisis-induced excess disengagement. The
co-movement pattern in Section~\ref{sec:germany} is therefore best read
as a qualitative implication of the model, not as an identification
result. Precise estimation would require a stochastic state-space
formulation and richer time-series or panel data.

In the symmetric shock specification, each crisis episode has a state
component $\Delta_k$ and a structural component $\Delta\beta_k$. The state
component can produce a temporary turnout disruption if the election
occurs before the disengagement reservoir has been depleted; the
structural component is reflected in the post-shock long-run floor
$V_k^*$. Disentangling the two components within a single election cycle
requires timing assumptions or auxiliary data on between-election
attitudes, turnout intentions, and voter transitions.

On the mathematical side, the main open problem is a sharper transient theory
for the asymmetric four-group model, including explicit bounds on surge
duration and peak size under repeated asymmetric shocks, and a more explicit
global theory for structural-shock sequences with parameter drift.  Stochastic
forcing and time-periodic shocks would further allow the study of threshold
crossing and ratchet effects outside the autonomous setting.

The main methodological message is that reversible surges and lasting
political shifts are distinct dynamical phenomena and should not be
attributed to the same mechanism: pure state shocks govern transient
amplification, whereas structural parameter shifts govern long-run
realignment.

\appendix
\section{Mathematical Proofs: Baseline Model}\label{app:proofs}

This appendix collects the proofs of all results stated in
Sections~\ref{sec:model} and~\ref{sec:baseline}, together with supporting
technical material on the bifurcation structure and comparative statics of
the baseline model.

\subsection*{A.1\quad Forward invariance}

\begin{proposition}[Forward invariance of $\Simp$]\label{prop:inv}
For any $(L_0,R_0)\in\Simp$ and any positive parameters, the solution
$(L(t),R(t))$ of~\eqref{eq:model} remains in $\Simp$ for all $t\ge0$.
\end{proposition}

\begin{proof}
We verify that the vector field $\mathbf{F}=(F_1,F_2)$ points into $\Simp$
(or is tangent to it) at every point on the boundary $\partial\Simp$.

\emph{Face $L=0$, $R\in[0,1]$:}
$F_1\big|_{L=0} = \gamma_{RL}\,R\,(1-R)\ge0$.
Hence $\dot{L}\ge0$: the trajectory cannot exit through this face.

\emph{Face $R=0$, $L\in[0,1]$:}
By symmetry, $F_2\big|_{R=0} = \gamma_{LR}\,L\,(1-L)\ge0$.

\emph{Face $C=0$, i.e.\ $L+R=1$:}
Here $C=0$, so $F_1\big|_{C=0}=-\mu_L\,L\le0$ and
$F_2\big|_{C=0}=-\mu_R\,R\le0$. Therefore $\dot{L}+\dot{R}\le0$, which means
$\dot{C}=-\dot{L}-\dot{R}\ge0$: the total radical share cannot increase
beyond~1.

Since the vector field points inward (or is tangent) on all three faces, and
the system has a unique locally Lipschitz solution by the
Picard--Lindel\"{o}f theorem, $\Simp$ is positively invariant.
\end{proof}

\begin{remark}
On the boundary face $C=0$ (i.e.\ $L+R=1$), at least one of $L,R$ is
positive, so
\[
  \dot C=\mu_L L+\mu_R R>0
\]
strictly. The face $C=0$ is therefore inward-pointing: a trajectory with
$C(0)>0$ cannot reach $C=0$ in finite time, and a trajectory with $C(0)=0$
returns immediately to the interior.
\end{remark}

\subsection*{A.2\quad Symmetric reduction and global dynamics}

In the symmetric case $\alpha_L=\alpha_R=\alpha$, $\mu_L=\mu_R=\mu$,
$\gamma_{LR}=\gamma_{RL}=\gamma$, set $\beta:=\alpha+\gamma$ and introduce
the total radical share $S:=L+R$ and the left--right imbalance $D:=L-R$.
Since $C=1-S$, a direct computation yields the triangular system
\begin{equation}\label{eq:SD_system}
  \dot{S} = S\bigl[\beta(1-S)-\mu\bigr],
  \qquad
  \dot{D} = D\bigl[(\alpha-\gamma)(1-S)-\mu\bigr].
\end{equation}
The $S$-equation is logistic with unique positive equilibrium
$S^*=1-\mu/\beta$ when $\beta>\mu$. On the diagonal $\{L=R\}$, setting
$S=2P$ gives
\begin{equation}\label{eq:scalar}
  \dot{P} = P[\beta(1-2P)-\mu].
\end{equation}

\begin{theorem}[Global dynamics of the symmetric model]\label{thm:global}
\begin{enumerate}
  \item If $\beta\le\mu$: $E_0=(0,0)$ is globally asymptotically stable on
    $\Simp$. Both radical wings vanish: $L(t),R(t)\to0$, $C(t)\to1$.
  \item If $\beta>\mu$: the unique interior equilibrium
    $E_1=(P^*,P^*)$, with $P^*=\tfrac12(1-\mu/\beta)$ and $C^*=\mu/\beta$,
    is globally asymptotically stable on $\Simp\setminus\{E_0\}$; $E_0$ is
    unstable.
\end{enumerate}
\end{theorem}

\begin{proof}
\textbf{Case $\beta\le\mu$:}
For $S\ge0$, $\dot S=(\beta-\mu)S-\beta S^2\le(\beta-\mu)S$. If
$\beta<\mu$, then $S(t)\le S(0)e^{(\beta-\mu)t}\to0$. If $\beta=\mu$,
then $\dot S=-\beta S^2\le0$ with equality only at $S=0$, so $S(t)\to0$
by integration of the scalar logistic equation. Since $|D|\le S$, also
$D(t)\to0$, hence $L,R\to0$.

\textbf{Case $\beta>\mu$:}
The $S$-equation is logistic; for any $S(0)\in(0,1]$, $S(t)\to
S^*=1-\mu/\beta$ monotonically. The $D$-equation is linear:
\[
  \dot D = k(t)\,D,\qquad k(t):=(\alpha-\gamma)(1-S(t))-\mu.
\]
Since $S(t)\to S^*$, we have $k(t)\to(\alpha-\gamma)\mu/\beta-\mu=
-2\gamma\mu/\beta<0$. Hence there exist $T>0$ and $\eta>0$ such that
$k(t)\le-\eta$ for all $t\ge T$, giving $|D(t)|\le|D(T)|e^{-\eta(t-T)}\to0$.
Consequently $L(t)\to S^*/2=P^*$ and $R(t)\to P^*$.
Instability of $E_0$: its dominant eigenvalue is $\lambda_1=\beta-\mu>0$.
\end{proof}

\begin{corollary}[Mainstream share at equilibrium]\label{cor:Cstar}
When $E_1$ exists ($\beta>\mu$), the mainstream share at equilibrium is
\[
  C^* = \frac{\mu}{\beta} = \frac{\mu}{\alpha+\gamma}\in(0,1).
\]
$C^*$ is strictly decreasing in $\alpha$ and $\gamma$, strictly increasing in
$\mu$, and positive for all positive parameters.
\end{corollary}

\begin{proof}
Direct substitution of $P^*$ into $C=1-2P^*$.
\end{proof}

\subsection*{A.3\quad Asymmetric case: existence, uniqueness, and stability}

We first record an equivalence used repeatedly below.

\begin{lemma}[Perron threshold and spectral abscissa]\label{lem:pf_abscissa}
Let $K$ be a $2\times2$ entrywise positive matrix and let
$M=\mathrm{diag}(\mu_L,\mu_R)$ with $\mu_L,\mu_R>0$. Then
\[
  \lambda_{\mathrm{PF}}(M^{-1}K)\lessgtr1
  \quad\Longleftrightarrow\quad
  s(K-M)\lessgtr0,
\]
where $s(\cdot)$ denotes the spectral abscissa.
\end{lemma}

\begin{proof}
This is the nonsingular $M$-matrix criterion applied to $M-K$. Since $K$ is
entrywise positive and $M$ is a positive diagonal matrix, $M-K$ is a
nonsingular $M$-matrix iff $\lambda_{\mathrm{PF}}(M^{-1}K)<1$, equivalently
the Metzler matrix $K-M=-(M-K)$ has $s(K-M)<0$. The threshold and
supercritical cases follow by continuity and irreducibility of $K$.
\end{proof}

At an interior equilibrium $(L^*,R^*)$ with $C^*=1-L^*-R^*$, the equilibrium
conditions are
\begin{equation}\label{eq:eq_system}
\begin{aligned}
  C^*\bigl(\alpha_L L^*+\gamma_{RL}R^*\bigr) &= \mu_L L^*,\\
  C^*\bigl(\alpha_R R^*+\gamma_{LR}L^*\bigr) &= \mu_R R^*.
\end{aligned}
\end{equation}

The Jacobian of~\eqref{eq:model} at a general point $(L,R)$ is
\begin{equation}\label{eq:jacobian}
  J=
  \begin{pmatrix}
    \alpha_L(C-L)-\mu_L-\gamma_{RL}R &
    \gamma_{RL}(C-R)-\alpha_L L \\[3pt]
    \gamma_{LR}(C-L)-\alpha_R R &
    \alpha_R(C-R)-\mu_R-\gamma_{LR}L
  \end{pmatrix},
  \qquad C=1-L-R.
\end{equation}

\begin{theorem}[Existence, uniqueness, and local stability]\label{thm:PF}
Define
\[
  K:=\begin{pmatrix}\alpha_L & \gamma_{RL}\\\gamma_{LR} & \alpha_R\end{pmatrix},
  \quad
  M:=\begin{pmatrix}\mu_L & 0\\ 0 & \mu_R\end{pmatrix},
\]
and let $\lambda_{\mathrm{PF}}>0$ be the Perron root of $M^{-1}K$, with
explicit formula
\[
  \lambda_{\mathrm{PF}}=
  \frac{1}{2}\!\left(\frac{\alpha_L}{\mu_L}+\frac{\alpha_R}{\mu_R}
  +\sqrt{\!\left(\frac{\alpha_L}{\mu_L}-\frac{\alpha_R}{\mu_R}\right)^{\!2}
  +\frac{4\gamma_{RL}\gamma_{LR}}{\mu_L\mu_R}}\,\right).
\]
Let $u=(u_1,u_2)^\top\gg0$ be the corresponding Perron eigenvector.
\begin{enumerate}
  \item An interior equilibrium exists if and only if $\lambda_{\mathrm{PF}}>1$.
  \item When it exists it is unique, with
    $C^*=1/\lambda_{\mathrm{PF}}$ and
    $(L^*,R^*)=\tfrac{1-C^*}{u_1+u_2}(u_1,u_2)$.
  \item It is locally asymptotically stable.
\end{enumerate}
Note that $\mathcal{R}_{\mathrm{base}}=\lambda_{\mathrm{PF}}$, so
Theorem~\ref{thm:baseline_threshold} is the global counterpart of this local
result.
\end{theorem}

\begin{proof}
\textbf{Parts 1--2.} At an interior equilibrium,
equations~\eqref{eq:eq_system} read $C^*Kv=Mv$ for $v=(L^*,R^*)^\top\gg0$. Multiplying by $M^{-1}$ gives
$(M^{-1}K)v=(1/C^*)v$. By the Perron--Frobenius theorem, the only positive
eigenvector of $M^{-1}K$ is the Perron eigenvector with eigenvalue
$\lambda_{\mathrm{PF}}$, so $C^*=1/\lambda_{\mathrm{PF}}\in(0,1)$ iff
$\lambda_{\mathrm{PF}}>1$. The explicit formula follows from the $2\times2$
characteristic polynomial of $M^{-1}K$.

\textbf{Part 3.} At the interior equilibrium the equilibrium relations from
\eqref{eq:eq_system} give
\[
  J_{11}=-\alpha_LL^*-\gamma_{RL}R^*-C^*\gamma_{RL}R^*/L^*,\quad
  J_{22}=-\alpha_RR^*-\gamma_{LR}L^*-C^*\gamma_{LR}L^*/R^*.
\]
Hence $\mathrm{tr}\,J<0$. A direct computation using the equilibrium
relations yields $\det J=C^*(L^*+R^*)(\alpha_L\gamma_{LR}L^*/R^*
+2\gamma_{LR}\gamma_{RL}+\alpha_R\gamma_{RL}R^*/L^*)>0$. Thus $E_1$ is
locally asymptotically stable.
\end{proof}

\subsection*{A.4\quad Eigenvalue structure in the symmetric case}

\begin{theorem}[Eigenvalues in the symmetric 2D case]\label{thm:2d}
In the symmetric model ($\beta>\mu$), $E_1=(P^*,P^*)$ is locally
asymptotically stable with eigenvalues
\[
  \lambda_1 = \mu-\beta<0,\qquad \lambda_2 = -\frac{2\gamma\mu}{\beta}<0.
\]
Generically ($\lambda_1\ne\lambda_2$) it is a stable node. The eigenvalues
coincide on the codimension-one locus $2\gamma\mu=\beta(\beta-\mu)$, at
which $E_1$ is a stable star.
\end{theorem}

\begin{proof}
Evaluate the Jacobian at $L=R=P^*$ with $C^*=\mu/\beta$ and
$P^*=(\beta-\mu)/(2\beta)$. The entries are
\[
  a:=J_{11}=J_{22}=\frac{\mu(\alpha-\gamma)-\beta^2}{2\beta},\qquad
  b:=J_{12}=J_{21}=\frac{2\gamma\mu-\beta(\beta-\mu)}{2\beta}.
\]
The matrix $\bigl(\begin{smallmatrix}a&b\\b&a\end{smallmatrix}\bigr)$ has
eigenvalues $\lambda_{1,2}=a\pm b$, which simplify to $\lambda_1=\mu-\beta$
and $\lambda_2=-2\gamma\mu/\beta$.
\end{proof}

\begin{remark}
The eigenvalue $\lambda_1=\mu-\beta$ governs the symmetric mode in which
$L$ and $R$ change together. The eigenvalue $\lambda_2=-2\gamma\mu/\beta$
governs the asymmetric mode, i.e.\ the difference $L-R$. In the limit
$\gamma\to0$, $\lambda_2\to0$: without reactive cross-coupling the model
no longer selects a unique left--right split at fixed total radical share.
The supercritical $\gamma=0$ system instead admits a continuum of equilibria
on the segment $\{(L,R):L,R\ge0,\;L+R=1-\mu/\alpha\}$.
\end{remark}

\subsection*{A.5\quad Absence of periodic orbits}

\begin{corollary}[No periodic orbits]\label{cor:dulac}
The system~\eqref{eq:model} has no closed orbits in the interior of $\Simp$
for any positive parameter values.
\end{corollary}

\begin{proof}
Apply the Bendixson--Dulac criterion with weight $B(L,R)=1/(LRC)>0$ on
$\mathrm{int}(\Simp)$. Writing $BF_1=\alpha_L/R-\mu_L/(RC)+\gamma_{RL}/L$
and $BF_2=\alpha_R/L-\mu_R/(LC)+\gamma_{LR}/R$, a direct computation gives
\[
  \frac{\partial(BF_1)}{\partial L}+\frac{\partial(BF_2)}{\partial R}
  = -\frac{\gamma_{RL}}{L^2}-\frac{\gamma_{LR}}{R^2}
    -\frac{\mu_L}{RC^2}-\frac{\mu_R}{LC^2} < 0
\]
throughout $\mathrm{int}(\Simp)$. By Dulac's theorem there are no closed
orbits.
\end{proof}

\subsection*{A.6\quad Global dynamics: asymmetric case}

\begin{remark}[$E_0$ stability in the asymmetric model]
Evaluating the Jacobian at $E_0=(0,0)$ (where $C=1$) gives
\[
  J\big|_{E_0}=\begin{pmatrix}\alpha_L-\mu_L & \gamma_{RL}\\
  \gamma_{LR} & \alpha_R-\mu_R\end{pmatrix}=K-M.
\]
The radical-free equilibrium $E_0$ is locally asymptotically stable iff
$\alpha_L+\alpha_R<\mu_L+\mu_R$ and
$(\alpha_L-\mu_L)(\alpha_R-\mu_R)>\gamma_{RL}\gamma_{LR}$. By
Lemma~\ref{lem:pf_abscissa}, this is equivalent to
$\lambda_{\mathrm{PF}}(M^{-1}K)<1$; $E_0$ is nonhyperbolic at
$\lambda_{\mathrm{PF}}=1$ and unstable for $\lambda_{\mathrm{PF}}>1$. In
the symmetric case both characterisations reduce to $\beta<\mu$.
\end{remark}

\begin{theorem}[Global dynamics: full asymmetric model]\label{thm:global_asym}
Let $\lambda_{\mathrm{PF}}$ be as in Theorem~\ref{thm:PF}.
\begin{enumerate}
  \item If $\lambda_{\mathrm{PF}}\le1$, $E_0=(0,0)$ is globally
    asymptotically stable on $\Simp$.
  \item If $\lambda_{\mathrm{PF}}>1$, the unique interior equilibrium $E_1$
    is globally asymptotically stable on $\Simp\setminus\{E_0\}$.
\end{enumerate}
\end{theorem}

\begin{proof}
Write $x=(L,R)^\top$ and $S=L+R$. Since $C=1-S$,
\[
  \dot x=CKx-Mx=(K-M)x-SKx,
  \qquad
  N(x):=\dot x-(K-M)x=-SKx.
\]
On $\Simp$, $N(x)\le0$ componentwise, with $N(x)<0$ componentwise whenever
$x\ne0$. By Lemma~\ref{lem:pf_abscissa}, the sign of $s(K-M)$ coincides
with the sign of $\lambda_{\mathrm{PF}}-1$.

\textbf{Case $\lambda_{\mathrm{PF}}\le1$.}
Let $q\gg0$ be the positive left Perron vector of the Metzler matrix
$K-M$, so $q^\top(K-M)=sq^\top$ with $s:=s(K-M)\le0$. Set $W(x):=q^\top x$.
Then
\[
  \dot W=sW+q^\top N(x)\le0,
\]
with strict inequality whenever $x\ne0$. Global asymptotic stability of
$E_0$ on $\Simp$ follows by LaSalle's invariance principle.

\textbf{Case $\lambda_{\mathrm{PF}}>1$.}
Now $s:=s(K-M)>0$. Let $q\gg0$ again satisfy $q^\top(K-M)=sq^\top$, set
$W:=q^\top x$, and define
\[
  q_{\min}:=\min\{q_1,q_2\},
  \qquad
  c_K:=\max_{j=1,2}\frac{(q^\top K)_j}{q_j}.
\]
From $W=q_1L+q_2R\ge q_{\min}S$ we have $S\le W/q_{\min}$, and since
$(q^\top K)_j\le c_Kq_j$ entrywise we have $q^\top Kx\le c_KW$. Therefore
\[
  \dot W=sW-Sq^\top Kx
  \;\ge\;
  W\!\left(s-\frac{c_K}{q_{\min}}W\right).
\]
Set
\[
  \varepsilon:=\frac{sq_{\min}}{2c_K}>0.
\]
Whenever $0<W\le\varepsilon$, the bracket is $\ge s/2$, hence
$\dot W\ge(s/2)W>0$. Consequently $W$ cannot decrease past $\varepsilon$:
if $W(t_1)=\varepsilon$ at any time, then $\dot W(t_1)\ge\varepsilon s/2>0$,
contradicting downward crossing. In particular $\liminf_{t\to\infty}W(t)>0$
for every nonzero trajectory, so $E_0\notin\omega(x_0)$ for any
$x_0\in\Simp\setminus\{E_0\}$.

By Theorem~\ref{thm:PF}, the only equilibria in $\Simp$ are $E_0$ and the
unique interior equilibrium $E_1$. Boundary points other than $E_0$ are not
equilibria: on the open edges $\{L=0\}$ and $\{R=0\}$ the vector field
points into the interior ($\dot L|_{L=0}=\gamma_{RL}RC>0$ for $0<R<1$;
symmetrically for $R=0$), and on $\{C=0\}$, $\dot C=\mu_LL+\mu_RR>0$ since
$L+R=1$ forces at least one wing positive.

Let $\omega(x_0)$ denote the $\omega$-limit set of a trajectory with
$x_0\ne E_0$. It is nonempty, compact, connected, and forward invariant.
Any equilibrium it contains must be $E_1$. If $\omega(x_0)$ contained no
equilibrium, the Poincar\'e--Bendixson theorem \citep{perko2001} would
give a periodic orbit, ruled out by Corollary~\ref{cor:dulac}. Hence
$E_1\in\omega(x_0)$, and since $E_1$ is locally asymptotically stable
(Theorem~\ref{thm:PF}), $\omega(x_0)=\{E_1\}$. Therefore every trajectory
starting in $\Simp\setminus\{E_0\}$ converges to $E_1$.
\end{proof}

\subsection*{A.7\quad Bifurcation structure and comparative statics}

\begin{theorem}[Transcritical bifurcation]\label{thm:bif}
The equilibrium branches of the symmetric scalar model~\eqref{eq:scalar}
are
\[
  P=0,
  \qquad
  P=\frac{\beta-\mu}{2\beta},
\]
intersecting at $(P,\beta)=(0,\mu)$. For $\beta<\mu$, the nonzero branch
lies outside the feasible region $P\ge0$ and $P=0$ is the only feasible
equilibrium; it is stable. For $\beta>\mu$, the nonzero branch enters the
feasible region and is stable, while $P=0$ becomes unstable. The two
branches exchange stability at $\beta=\mu$, which is a transcritical
bifurcation.
\end{theorem}

\begin{proof}
Let $f(P,\beta)=P[\beta(1-2P)-\mu]=(\beta-\mu)P-2\beta P^2$. At
$(P,\beta)=(0,\mu)$,
\[
  f=0,\qquad f_P=0,\qquad f_{PP}=-4\mu\ne0,\qquad f_{P\beta}=1\ne0,
\]
which are the standard non-degeneracy conditions for a transcritical
bifurcation.
\end{proof}

\begin{corollary}[Comparative statics]\label{cor:cs}
At the interior equilibrium $E_1$ in the symmetric model
($\beta=\alpha+\gamma>\mu$),
\[
  \frac{\partial P^*}{\partial\alpha}
  =\frac{\mu}{2(\alpha+\gamma)^2}>0,
  \quad
  \frac{\partial P^*}{\partial\gamma}
  =\frac{\mu}{2(\alpha+\gamma)^2}>0,
  \quad
  \frac{\partial P^*}{\partial\mu}
  =-\frac{1}{2(\alpha+\gamma)}<0,
\]
\[
  \frac{\partial C^*}{\partial\alpha}
  =-\frac{\mu}{(\alpha+\gamma)^2}<0,
  \quad
  \frac{\partial C^*}{\partial\gamma}
  =-\frac{\mu}{(\alpha+\gamma)^2}<0,
  \quad
  \frac{\partial C^*}{\partial\mu}
  =\frac{1}{\alpha+\gamma}>0.
\]
Stronger recruitment, stronger reactive polarisation, and weaker
reabsorption each increase the long-run radical share. The mainstream
share $C^*=\mu/\beta$ is strictly positive for all finite positive
parameters; mainstream collapse requires $\mu\to0$ or
$\alpha+\gamma\to\infty$.
\end{corollary}

\begin{proof}
Direct differentiation of $P^*=\tfrac12(1-\mu/\beta)$ and $C^*=\mu/\beta$
with $\beta=\alpha+\gamma$.
\end{proof}

\section{Mathematical Proofs: Extended Model}\label{app:proofs4}

This appendix collects the proofs of all results stated in
Sections~\ref{sec:extended} and~\ref{sec:legacy}.

\subsection*{B.1\quad Forward invariance of $\mathcal{T}$}

\begin{proposition}[Forward invariance of $\mathcal{T}$]\label{prop:inv4}
For any $(L_0,R_0,A_0)\in\mathcal{T}$, positive constant parameters, and
any locally bounded piecewise continuous nonnegative crisis input
$\sigma\colon[0,\infty)\to[0,\infty)$, the Carath\'eodory solution of
\eqref{eq:4group} is unique and remains in $\mathcal{T}$ for all
$t\ge0$. The same invariance conclusion holds for locally integrable
nonnegative inputs in the Carath\'eodory sense. In particular, this
holds for the post-shock autonomous system with $\sigma\equiv0$.
Moreover, the impulse map in Definition~\ref{def:impulse} maps
$\mathcal{T}$ into itself.
\end{proposition}

\begin{proof}
We verify that the vector field points inward (or is tangent) on each face
of $\mathcal{T}$.

\emph{Face $L=0$:} $\dot{L}\big|_{L=0}=\gamma_{RL}\,R\,C\ge0$.

\emph{Face $R=0$:} $\dot{R}\big|_{R=0}=\gamma_{LR}\,L\,C\ge0$.

\emph{Face $A=0$:} $\dot{A}\big|_{A=0}=\sigma(t)\,C\ge0$
(with $\sigma\equiv0$ post-shock, $\dot{A}\big|_{A=0}=0$).

\emph{Face $C=0$ (i.e.\ $L+R+A=1$):} Here $C=0$.
Then $\dot{L}\big|_{C=0}=\delta_L AL-\mu_L L=L(\delta_L A-\mu_L)$ and
similarly for $R$. Also $\dot{A}\big|_{C=0}=-A(\delta_L L+\delta_R R+\rho)\le0$.
Therefore
\begin{align*}
  \dot{L}+\dot{R}+\dot{A}\big|_{C=0}
  &= L(\delta_L A-\mu_L)+R(\delta_R A-\mu_R)-A(\delta_L L+\delta_R R+\rho)\\
  &= -\mu_L L-\mu_R R-\rho A\le0.
\end{align*}
Hence $\dot{C}=-(\dot{L}+\dot{R}+\dot{A})\ge0$: the total $L+R+A$ cannot
increase beyond~1.

For locally bounded piecewise continuous $\sigma$, the vector field is
measurable in time and locally Lipschitz in the state variables on each
compact subset of $\mathcal{T}$. Standard Carath\'eodory existence and
uniqueness therefore gives a unique absolutely continuous solution. The
inward-pointing boundary inequalities above imply positive invariance of
$\mathcal{T}$; the same argument applies to locally integrable
nonnegative inputs.

The impulse map in Definition~\ref{def:impulse} preserves nonnegativity and
total mass: $C^+=(1-\Delta)C^-\ge0$, $A^+=A^-+\Delta C^-\ge0$, and $L,R$ are
unchanged, with $L^++R^++C^++A^+=L^-+R^-+C^-+A^-=1$.
\end{proof}

\subsection*{B.2\quad Equilibrium inheritance and decay of disengagement}

\begin{proposition}[Equilibrium inheritance from the baseline]\label{prop:inherit4}
For the full four-group model \eqref{eq:4group} with $\sigma\equiv0$ and all
parameters positive, every equilibrium has $A^*=0$. Consequently, once the
crisis input has ended, the disengagement compartment does not create any
additional equilibria: the equilibria of the four-group system are exactly
the baseline equilibria embedded in the invariant face $A=0$.

Moreover, at any such equilibrium the Jacobian is block upper-triangular:
\[
  J_4(L^*,R^*,0)=
  \begin{pmatrix}
    J_{\mathrm{base}}(L^*,R^*) & *\\
    0\;\;0 & -(\delta_L L^*+\delta_R R^*+\rho)
  \end{pmatrix},
\]
so local stability is inherited from the baseline model together with one
additional strictly negative transverse eigenvalue
$-(\delta_L L^*+\delta_R R^*+\rho)<0$.
\end{proposition}

\begin{proof}
At any equilibrium, $0=\dot{A}=-A(\delta_L L+\delta_R R+\rho)$. Since
$\delta_L,\delta_R,\rho>0$, the coefficient is strictly positive, hence
$A^*=0$. Substituting $A=0$ reduces \eqref{eq:4group} to exactly the baseline
equations \eqref{eq:model}. The block-triangular Jacobian follows by direct
differentiation.
\end{proof}

\begin{proposition}[Exponential decay of the disengaged compartment]
\label{prop:Adecay4}
For every solution of the post-shock autonomous system \eqref{eq:4group}
with $\sigma\equiv0$,
\[
  A(t)=A(0)\exp\!\left(-\int_0^t(\delta_L L(s)+\delta_R R(s)+\rho)\,ds\right)
  \le A(0)e^{-\rho t}.
\]
In particular, $A(t)\to0$ exponentially as $t\to\infty$.
\end{proposition}

\begin{proof}
The equation $\dot A=-(\delta_L L+\delta_R R+\rho)A$ is scalar and separable.
Integrating from $0$ to $t$ gives the exact formula. Since
$\delta_L L+\delta_R R+\rho\ge\rho>0$, the exponential bound follows
immediately.
\end{proof}

\subsection*{B.3\quad Global post-shock dynamics}

\begin{proof}[Proof of Theorem~\ref{thm:global4}]
Write $x=(L,R)^\top$ and $D:=\mathrm{diag}(\delta_L,\delta_R)$, and let
$s(\cdot)$ denote the spectral abscissa. By the Perron--Frobenius/M-matrix
equivalence,
\[
  \mathcal R_{\mathrm{rad}}\le1\iff s(K-M)\le0,
  \qquad
  \mathcal R_{\mathrm{rad}}>1\iff s(K-M)>0.
\]
Using $C=1-L-R-A$, the $(L,R)$-subsystem is
\[
\begin{aligned}
  \dot x
  &=[(1-A)K+AD-M]x-(L+R)Kx\\
  &=(K-M)x+A(D-K)x-(L+R)Kx.
\end{aligned}
\]

\textbf{Part (i).} Let $q\gg0$ satisfy $q^\top(K-M)=s_0 q^\top$ with
$s_0:=s(K-M)\le0$. Choose
\[
  \eta>\max\!\left\{0,\frac{(q^\top(D-K))_1}{\delta_L},
  \frac{(q^\top(D-K))_2}{\delta_R}\right\}
\]
and set $W(L,R,A):=q^\top x+\eta A$. Differentiating:
\begin{align*}
  \dot W
  &=s_0(q^\top x)+A\,q^\top(D-K)x-(L+R)q^\top Kx-\eta(\delta_L L+\delta_R R+\rho)A\\
  &=s_0(q^\top x)-(L+R)q^\top Kx-\eta\rho A
    +A\bigl[(q^\top(D-K))_1-\eta\delta_L\bigr]L
    +A\bigl[(q^\top(D-K))_2-\eta\delta_R\bigr]R.
\end{align*}
By construction of $\eta$ the last two terms are nonpositive. Since $s_0\le0$,
$q\gg0$, and $K>0$ entrywise, $\dot W<0$ for every $(L,R,A)\ne(0,0,0)$.
Hence $W$ is a strict Lyapunov function and $E_0$ is globally asymptotically
stable.

\textbf{Part (ii).} On $\Gamma$: $L=R=0$ gives $\dot A=-\rho A$, so $\Gamma$
is forward invariant and trajectories on it converge to $E_0$.

Now let $(L_0,R_0,A_0)\in\mathcal T\setminus\Gamma$. By uniqueness and
invariance of $\Gamma$, $x(t)\ne0$ for all $t\ge0$, so
$W(t):=q^\top x(t)>0$. Since $\mathcal R_{\mathrm{rad}}>1$, we have
$s_0>0$.

By Proposition~\ref{prop:Adecay4}, $A(t)\to0$ exponentially. Define
$\kappa:=\max_i|(q^\top(D-K))_i|/q_i$, $c_K:=\max_i(q^\top K)_i/q_i$,
$q_{\min}:=\min\{q_1,q_2\}$. Then
\[
  \dot W\ge\left[s_0-\kappa A(t)-\frac{c_K}{q_{\min}}W\right]W.
\]
Choose $T>0$ such that $\kappa A(t)\le s_0/4$ for $t\ge T$, and set
$\varepsilon:=s_0 q_{\min}/(4c_K)>0$. For $t\ge T$ and $0<W(t)\le\varepsilon$:
$\dot W(t)\ge(s_0/2)W(t)>0$. It follows that there exists $T_1\ge T$ with
$W(t)\ge\varepsilon$ for all $t\ge T_1$.

Therefore every $\omega$-limit set of a trajectory in $\mathcal T\setminus\Gamma$
lies in $\Sigma_2^\varepsilon:=\{(L,R,0)\in\mathcal T:q^\top(L,R)^\top\ge\varepsilon\}
\subset\Sigma_2\setminus\{E_0\}$, where $\Sigma_2:=\{(L,R,0)\in\mathcal T\}$
is the baseline face. Since $A(t)\to0$ exponentially (Proposition~\ref{prop:Adecay4}),
every $\omega$-limit point satisfies $A=0$; the restriction to the face $A=0$
coincides with the baseline flow \eqref{eq:model}. By
Theorem~\ref{thm:global_asym}, every $\omega$-limit set in
$\Sigma_2^\varepsilon\subset\Sigma_2\setminus\{E_0\}$ is the singleton
$\{E_1\}$.

Hence $(L(t),R(t),A(t))\to E_1$ for every initial condition in
$\mathcal T\setminus\Gamma$.
\end{proof}

\subsection*{B.4\quad Perron--Frobenius threshold in the asymmetric model}

\begin{remark}[Relation to next-generation thresholds]\label{rem:ngm}
The quantity $\mathcal R_{\mathrm{rad}}:=\lambda_{\mathrm{PF}}(M^{-1}K)$ is
the political analogue of a basic reproduction number. In epidemiological
compartment models, thresholds of this type arise from the next-generation
matrix framework of \citet{diekmann1990} and \citet{vandendriessche2002}:
the matrix $M^{-1}K$ plays the role of the next-generation operator, $K$
collecting recruitment and reactive-polarisation terms and $M$ collecting
reabsorption terms. The threshold $\mathcal R_{\mathrm{rad}}=1$ has the
same mathematical status as $R_0=1$: below it the mainstream equilibrium is
stable; above it radical persistence becomes possible.
\end{remark}

\begin{theorem}[Perron--Frobenius threshold in the full asymmetric model]
\label{thm:PF4}
Consider the post-shock system \eqref{eq:4group} with $\sigma\equiv0$ and
post-shock parameter vector $\theta^+$. Define $K^+$, $M^+$,
$B^+:=(M^+)^{-1}K^+$, and $\mathcal R_{\mathrm{rad}}^+:=\lambda_{\mathrm{PF}}(B^+)$.
Then:
\begin{enumerate}
  \item $E_0=(0,0,0)$ is locally asymptotically stable iff
    $\mathcal R_{\mathrm{rad}}^+<1$, nonhyperbolic at $\mathcal R_{\mathrm{rad}}^+=1$,
    and unstable if $\mathcal R_{\mathrm{rad}}^+>1$.
\item A nontrivial equilibrium with $L^*>0$, $R^*>0$ exists iff
    $\mathcal R_{\mathrm{rad}}^+>1$. It is unique: if $v^+\gg0$ is the
    right Perron--Frobenius eigenvector of $B^+$ normalised by
    $v_L^++v_R^+=1$, then
\[
      (L^*,R^*)=\left(1-\frac{1}{\mathcal R_{\mathrm{rad}}^+}\right)v^+,
      \quad
      C^*=\frac{1}{\mathcal R_{\mathrm{rad}}^+},
      \quad
      A^*=0.
    \]
    Local stability is inherited from Theorem~\ref{thm:PF} via
    Proposition~\ref{prop:inherit4}.
  \item In the symmetric reduction,
    $\mathcal R_{\mathrm{rad}}^+=(\alpha^++\gamma^+)/\mu^+=\beta^+/\mu^+$.
\end{enumerate}
\end{theorem}

\begin{proof}
At $E_0$ with $C=1$, $A=0$, the Jacobian of \eqref{eq:4group} is
$J_4(E_0)=\mathrm{blockdiag}(K^+-M^+,\,-\rho^+)$.
Local stability of $E_0$ is equivalent to $s(K^+-M^+)<0$, which by the
M-matrix criterion is equivalent to $\mathcal R_{\mathrm{rad}}^+<1$. Parts (ii)
and (iii) follow from Proposition~\ref{prop:inherit4} and Theorem~\ref{thm:PF}.
\end{proof}

\begin{remark}[Relation to the symmetric shock results]\label{rem:PF4bridge}
The scalar threshold $\beta=\mu$ used in Theorems~\ref{thm:trigger}
and~\ref{thm:cumulative} is the symmetric special case of
$\mathcal R_{\mathrm{rad}}^+=1$. The location of the structural bifurcation is
determined by the recruitment, reactive-polarisation, and reabsorption
parameters $(\alpha_i,\gamma_{ij},\mu_i)$; the mobilisation and
re-engagement parameters $(\delta_i,\rho)$ govern transient amplification
and decay only.
\end{remark}

\subsection*{B.5\quad Symmetric stability and the mobilisation window}

In the symmetric parameter regime
\[
  \alpha_L=\alpha_R=\alpha,\qquad
  \mu_L=\mu_R=\mu,\qquad
  \gamma_{LR}=\gamma_{RL}=\gamma,\qquad
  \delta_L=\delta_R=\delta,
\]
write $S(t):=L(t)+R(t)$ and $\beta:=\alpha+\gamma$. The total radical share
and the disengagement compartment satisfy the closed two-dimensional system
\[
\begin{aligned}
  \dot S &= S\bigl[(\beta-\mu)-\beta S+(\delta-\beta)A\bigr],\\
  \dot A &= -(\delta S+\rho)A,
\end{aligned}
\]
on $\mathcal F=\{(S,A):\ S\ge0,\ A\ge0,\ S+A\le1\}$.

\begin{proposition}[Stability of equilibria in the symmetric model]
\label{prop:stab4}
For the symmetric two-dimensional $(S,A)$-system:
\begin{enumerate}
  \item If $\beta\le\mu$, the origin $(S,A)=(0,0)$ is globally
  asymptotically stable on $\mathcal F$.
  \item If $\beta>\mu$, the edge $\{S=0\}$ is forward invariant. Every
  trajectory with $S_0=0$ converges to the origin; every trajectory with
  $S_0>0$ converges to
  \[
    (S^*,A^*)=\left(1-\frac{\mu}{\beta},\,0\right).
  \]
  In the full symmetric $(L,R,A)$ system, this corresponds to
  $L^*=R^*=\tfrac12(1-\mu/\beta)$, $A^*=0$. In this regime the origin is
  a saddle and $(S^*,0)$ is a stable node.
\end{enumerate}
\end{proposition}

\begin{proof}
\textbf{Case (i).} Let $W(S,A):=S+A$. Then
\[
  \dot W=(\beta-\mu)S-\beta S^2-\beta SA-\rho A.
\]
If $\beta\le\mu$, every term is nonpositive and equality holds only at
$(S,A)=(0,0)$. LaSalle's invariance principle gives global asymptotic
stability of the origin.

\textbf{Case (ii).} If $S_0=0$, then $\dot S|_{S=0}=0$, so $S(t)\equiv0$
and $A(t)=A_0e^{-\rho t}\to0$.

If $S_0>0$, positivity gives $S(t)>0$ for all $t$. Since
$\dot A=-(\delta S+\rho)A$, we have $A(t)\to0$ exponentially.

Fix $\eta\in(0,\beta-\mu)$. For all sufficiently large $t$,
$|(\delta-\beta)A(t)|\le\eta$, hence
\[
  S[(\beta-\mu)-\eta-\beta S]\le\dot S\le S[(\beta-\mu)+\eta-\beta S].
\]
Comparison with the logistic equations
$\dot z_\pm=z_\pm[(\beta-\mu)\pm\eta-\beta z_\pm]$ gives
\[
  \frac{\beta-\mu-\eta}{\beta}
  \le\liminf_{t\to\infty}S(t)
  \le\limsup_{t\to\infty}S(t)
  \le\frac{\beta-\mu+\eta}{\beta}.
\]
Letting $\eta\downarrow0$ yields $S(t)\to1-\mu/\beta$.

The Jacobian in $(S,A)$ is upper triangular at both equilibria: at the
origin the eigenvalues are $\beta-\mu>0$ and $-\rho<0$ (saddle); at
$(S^*,0)$ they are $-(\beta-\mu)<0$ and $-(\delta S^*+\rho)<0$ (stable
node).
\end{proof}

\begin{remark}[Exact instantaneous growth criterion]\label{rem:Ac}
Assume $\delta>\beta$. For a general post-shock state $(S_0,A_0)$, the
exact condition for initial growth of total radical support is
\[
  \dot S(0)>0\iff A_0>A_c(S_0):=\frac{\mu-\beta+\beta S_0}{\delta-\beta}.
\]
Theorem~\ref{thm:critical} is the near-mainstream limit $S_0\to0$, giving
$A_c(S_0)\to\Delta_c$. For a system already substantially radicalised, the
effective threshold $A_c(S_0)>\Delta_c$ is higher.
\end{remark}

\begin{corollary}[Properties of $\Delta_c$]\label{cor:Dc}
The threshold $\Delta_c=(\mu-\beta)/(\delta-\beta)$ satisfies $\Delta_c>0$
iff $\mu>\beta$ and $\delta>\beta$, and $\Delta_c<1$ iff $\delta>\mu$. The
politically meaningful feasible surge regime is therefore
$\delta>\mu>\beta$. Moreover,
\[
  \frac{\partial\Delta_c}{\partial\mu}>0,
  \qquad
  \frac{\partial\Delta_c}{\partial\delta}<0,
  \qquad
  \frac{\partial\Delta_c}{\partial\beta}
  =\frac{\mu-\delta}{(\delta-\beta)^2}<0
  \text{ when }\delta>\mu,
\]
and $\Delta_c\to0$ as $\beta\to\mu^-$.
\end{corollary}

\begin{proof}
Direct differentiation and sign analysis. The condition $\Delta_c<1$ is
$\mu-\beta<\delta-\beta$, i.e.\ $\delta>\mu$.
\end{proof}

\begin{proposition}[Mobilisation-window bound]\label{prop:window4}
Assume $\beta<\mu<\delta$ and let $\Delta_c=(\mu-\beta)/(\delta-\beta)$.
Let $A_0=A(t_0^+)$ be the disengagement level immediately after a state
shock. If $A_0>\Delta_c$, then for $\tau:=t-t_0$ and
\[
  \tau^*:=\frac{1}{\rho}\ln\!\left(\frac{A_0}{\Delta_c}\right),
\]
one has $\dot S(\tau)<0$ for every $\tau\ge\tau^*$ with $S(\tau)>0$. In
particular, the surge window is contained in $[0,\tau^*]$. In the
near-mainstream limit, $A_0\to\Delta$, giving the bound of
Theorem~\ref{thm:window}.
\end{proposition}

\begin{proof}
For $S>0$,
$\dot S/S=(\beta-\mu)-\beta S+(\delta-\beta)A\le(\beta-\mu)+(\delta-\beta)A$.
Since $\dot A=-(\delta S+\rho)A\le-\rho A$, we have
$A(\tau)\le A_0e^{-\rho\tau}$. For $\tau\ge\tau^*$,
$A_0e^{-\rho\tau}\le\Delta_c$, hence
\[
  \frac{\dot S}{S}\le(\beta-\mu)+(\delta-\beta)\Delta_c-\beta S=-\beta S<0
\]
whenever $S>0$.
\end{proof}

\subsection*{B.6\quad Structural trigger and cumulative shocks}

\begin{proof}[Proof of Theorem~\ref{thm:trigger}]
In this theorem the structural component is restricted to a scalar shift
$\beta_0\to\beta'=\beta_0+\Delta\beta$, with $\mu$ (and $\delta,\rho$)
fixed.

If $S(t_0^-)=0$, then the state shock leaves $S(t_0^+)=0$, the edge
$\{S=0\}$ is invariant, and the system returns to the origin
(equivalently $E_0=(0,0,0)$ in full coordinates).

If $S(t_0^-)>0$, the state shock preserves $S(t_0^+)=S(t_0^-)>0$ (only
$C$ and $A$ are affected). The post-shock system evolves with parameter
$\beta'$. By Proposition~\ref{prop:stab4}(ii), the long-run attractor is
the origin if $\beta'\le\mu$, and $(S^*,A^*)=(1-\mu/\beta',\,0)$ if
$\beta'>\mu$. Hence a permanent shift in the long-run regime occurs iff
$\beta_0+\Delta\beta>\mu$. When this holds,
\[
  C^\infty=1-S^*=\frac{\mu}{\beta_0+\Delta\beta},
\]
which depends only on $\beta'$ and $\mu$ and is therefore independent of
$\Delta$.
\end{proof}

\begin{proof}[Proof of Theorem~\ref{thm:cumulative}]
After the $k$-th shock, the post-shock structural parameter is
$B_k=\beta_0+\sum_{j=1}^k\Delta\beta_j$. By Theorem~\ref{thm:trigger},
conditional on a nonzero radical seed whenever the post-shock regime is
supercritical, the long-run mainstream share is
\[
  C_k^\infty=
  \begin{cases}
    1, & B_k\le\mu,\\[2mm]
    \mu/B_k, & B_k>\mu.
  \end{cases}
\]
Parts (i)--(iii) follow immediately from this formula and the assumption
$\Delta\beta_k\ge0$. Part (iv) is the definition of the first
threshold-crossing index,
\[
  k^*=\min\!\left\{k:\sum_{j=1}^k\Delta\beta_j>\mu-\beta_0\right\}.
\]
\end{proof}

\subsection*{B.7\quad Asymmetric state-shock threshold}

\begin{theorem}[Asymmetric critical shock threshold]\label{thm:PF_Delta}
Let $K$, $M$, $D$ be as in Section~\ref{sec:legacy} and, for
$\Delta\in[0,1]$, define
\[
\begin{aligned}
  G(\Delta) &:= (1-\Delta)K+\Delta D-M,\\
  \Phi(\Delta) &:= \lambda_{\mathrm{PF}}\bigl(M^{-1}((1-\Delta)K+\Delta D)\bigr).
\end{aligned}
\]
For $\Delta\in[0,1)$ let $q(\Delta)\gg0$ denote the positive left
Perron--Frobenius eigenvector of the irreducible Metzler matrix
$G(\Delta)$ associated with $s(G(\Delta))$. At the endpoint $\Delta=1$,
irreducibility may fail, so statements about growth are made directly
through $s(G(1))$ or by the one-sided limit $\Delta\uparrow1$. Assume
$\Phi(0)<1$.
\begin{enumerate}
  \item If $\Phi(1)\le1$, then $\Phi(\Delta)\le1$ for all $\Delta\in[0,1]$,
    with strict inequality on $[0,1)$. Hence no feasible pure state shock
    produces positive near-mainstream radical growth.
  \item If $\Phi(1)>1$, there is a unique threshold
$\Delta_c^{\mathrm{asym}}\in(0,1)$, defined as the right endpoint of
$\{\Delta\in[0,1]:\Phi(\Delta)\le1\}$, with
$\Phi(\Delta_c^{\mathrm{asym}})=1$, $\Phi(\Delta)\le1$ for
$\Delta\le\Delta_c^{\mathrm{asym}}$, and $\Phi(\Delta)>1$ for
$\Delta>\Delta_c^{\mathrm{asym}}$. If the crossing is nondegenerate, the
subcritical inequality is strict for $\Delta<\Delta_c^{\mathrm{asym}}$.
    For a near-mainstream post-shock state $x(t_0^+)=\varepsilon v$ with
    $v\ge0$, $v\ne0$, and $\varepsilon\ll1$, the weighted radical mass
    $V_\Delta:=q(\Delta)^\top x$ satisfies
    \[
      \dot V_\Delta(t_0^+)=s(G(\Delta))\,V_\Delta(t_0^+)+O(\varepsilon^2),
    \]
    so the sign of initial weighted radical growth is determined by
    $\Delta-\Delta_c^{\mathrm{asym}}$.
\end{enumerate}
In the symmetric case,
$\Delta_c^{\mathrm{asym}}=(\mu-\beta)/(\delta-\beta)$.
\end{theorem}

\begin{proof}
The sign criterion follows from
$q(\Delta)^\top\dot x=s(G(\Delta))\,q(\Delta)^\top x+O(\varepsilon^2)$
together with the Perron--Frobenius/M-matrix equivalence
$s(G(\Delta))\gtrless0\iff\Phi(\Delta)\gtrless1$. Writing
$a_L(\Delta)=((1-\Delta)\alpha_L+\Delta\delta_L)/\mu_L$,
$a_R(\Delta)=((1-\Delta)\alpha_R+\Delta\delta_R)/\mu_R$,
$b_0=\gamma_{RL}/\mu_L$, $c_0=\gamma_{LR}/\mu_R$:
\[
  \Phi(\Delta)=\tfrac12\bigl(a_L+a_R+\sqrt{(a_L-a_R)^2+4b_0c_0(1-\Delta)^2}\bigr).
\]
The square-root term is the Euclidean norm of an affine function of
$\Delta$, hence convex; therefore $\Phi$ is continuous and convex on
$[0,1]$. If $\Phi(1)\le1$, convexity gives $\Phi(\Delta)\le1$ on $[0,1]$,
strict on $[0,1)$. If $\Phi(1)>1$, the sublevel set
$\{\Delta:\Phi(\Delta)\le1\}$ is an interval containing $0$ but not $1$,
which gives a unique boundary point $\Delta_c^{\mathrm{asym}}\in(0,1)$.
\end{proof}

\begin{remark}[Quadratic formula for $\Delta_c^{\mathrm{asym}}$]\label{rem:Dc_formula}
Write $a_i:=\alpha_i-\mu_i$, $e_i:=\delta_i-\alpha_i$,
$g:=\gamma_{RL}\gamma_{LR}$. Setting $\det G(\Delta)=0$ yields the
quadratic
\begin{equation}\label{eq:Dc_quadratic}
  q_2\Delta^2+q_1\Delta+q_0=0,
\end{equation}
where
\[
  q_2:=e_Le_R-g,
  \quad
  q_1:=a_Le_R+a_Re_L+2g,
  \quad
  q_0:=a_La_R-g>0.
\]
When $q_2\ne0$, equation \eqref{eq:Dc_quadratic} has roots
\[
  \frac{-q_1\pm\sqrt{q_1^2-4q_2q_0}}{2q_2}.
\]
Under the assumption $\Phi(0)<1<\Phi(1)$, the critical value
$\Delta_c^{\mathrm{asym}}$ is the smallest root in $(0,1)$ at which the
Perron eigenvalue crosses zero; equivalently,
\[
  \Delta_c^{\mathrm{asym}}
  =\min\bigl\{\Delta\in(0,1):\det G(\Delta)=0\bigr\},
\]
which corresponds to the root with $\operatorname{tr}G(\Delta)\le0$. In
the symmetric case \eqref{eq:Dc_quadratic} reduces to
$(a+e\Delta)^2=\gamma^2(1-\Delta)^2$, giving
$\Delta_c^{\mathrm{asym}}=(\mu-\beta)/(\delta-\beta)=\Delta_c$.
\end{remark}

\begin{proposition}[Asymmetric Perron-vector window bound]\label{prop:PF_window}
Let $s_0:=s(K-M)<0$, $q\gg0$ with $q^\top(K-M)=s_0 q^\top$, and
$\bar\kappa:=\max_i(q^\top(D-K))_i/q_i$. For a pure state shock of
amplitude $\Delta$ with $A(0^-)=0$, write $\tau:=t-t_0$ and
$W(\tau):=q^\top x(\tau)$. Then
\[
  \dot W(\tau)\le[s_0+\bar\kappa A(\tau)]W(\tau)
  \le[s_0+\bar\kappa\Delta e^{-\rho\tau}]W(\tau).
\]
If $\bar\kappa\le0$, $W$ is strictly decreasing. If $\bar\kappa>0$,
define $\Delta_q:=-s_0/\bar\kappa$ and
$\tau_q^*:=\rho^{-1}\ln(\Delta/\Delta_q)$ for $\Delta>\Delta_q$. Then
$\Delta\le\Delta_q$ implies $W$ nonincreasing; if $\Delta>\Delta_q$,
every interval on which $W$ increases is contained in $[0,\tau_q^*)$.

In the symmetric case: $s_0=\beta-\mu$, $\bar\kappa=\delta-\beta$,
$\Delta_q=\Delta_c$, and $\tau_q^*$ coincides with the bound $\tau^*$ of
Theorem~\ref{thm:window}.
\end{proposition}

\begin{proof}
Using $C=1-L-R-A$, $\dot x=[(1-A)K+AD-M]x-(L+R)Kx$, so
\[
  \dot W=s_0 W+A\,q^\top(D-K)x-(L+R)q^\top Kx
  \le[s_0+\bar\kappa A(\tau)]W.
\]
Since $\dot A\le-\rho A$ and $A(0)\le\Delta$ from the jump condition,
$A(\tau)\le\Delta e^{-\rho\tau}$. The rest follows by sign analysis of
$s_0+\bar\kappa\Delta e^{-\rho\tau}$.
\end{proof}

\section{Additional Numerical Illustrations}\label{app:numerics}

All simulations use the \texttt{solve\_ivp} function with the RK45 method
(relative tolerance $10^{-10}$, absolute tolerance $10^{-12}$).

\subsection*{C.1\quad Baseline model}

Figure~\ref{fig:bif} shows the equilibrium radical share $P^*$ and mainstream
share $C^*$ as functions of the polarisation parameter $\gamma$ for three
values of the recruitment rate $\alpha$, with $\mu=0.20$. The transcritical
bifurcation at $\beta=\mu$ is clearly visible as the left endpoint of each
curve. The mainstream share $C^*=\mu/(\alpha+\gamma)$ approaches but never
reaches zero over the plotted range, consistent with Corollary~\ref{cor:cs}.

\begin{figure}[ht]
\centering
\includegraphics[width=\textwidth]{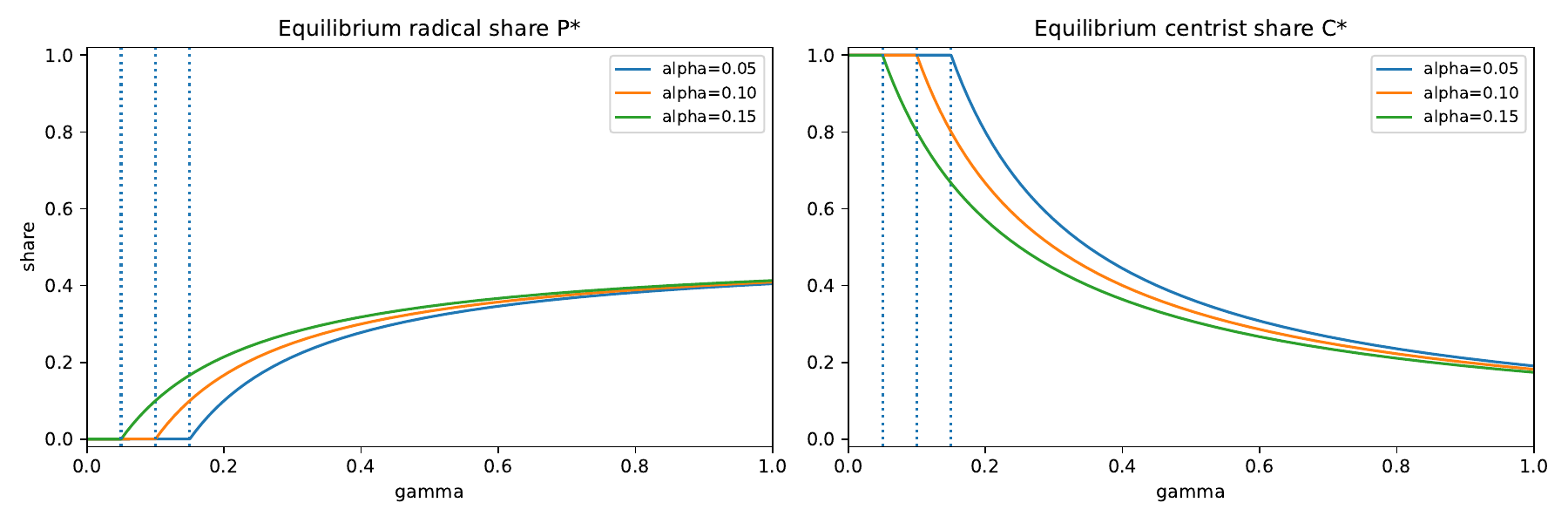}
\caption{Equilibrium radical share $P^*$ (left) and mainstream share
  $C^*$ (right) as functions of $\gamma$ for
  $\alpha\in\{0.05,0.10,0.15\}$ and $\mu=0.20$.}
\label{fig:bif}
\end{figure}

Figure~\ref{fig:phase} shows phase portraits in $\Sigma_2$ for the same
three parameter sets as Figure~\ref{fig:bif}. All trajectories converge
to the unique attractor (marked by stars), illustrating global stability
(Theorem~\ref{thm:baseline_threshold}); no limit cycles are present,
consistent with Corollary~\ref{cor:dulac}.

\begin{figure}[ht]
\centering
\includegraphics[width=\textwidth]{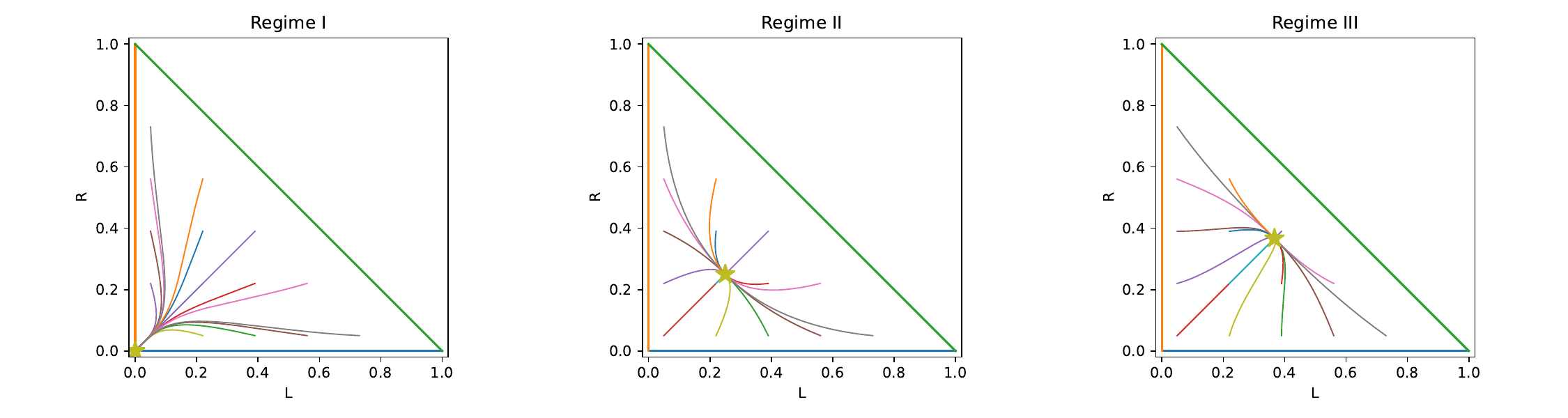}
\caption{Phase portraits in $\Sigma_2$ for the same three parameter sets
  as Figure~\ref{fig:bif}. Stars mark attracting equilibria.}
\label{fig:phase}
\end{figure}

\subsection*{C.2\quad Extended model}

The figures in this subsection use the symmetric extended model with
$L=R=:P$ and $\beta=\alpha+\gamma$. Trajectories are plotted in terms of
$P$ per radical wing; total radical support is therefore $S=2P$. Each
simulation starts with a small nonzero radical seed and $A(0)=0$, then
shocks are applied at the times stated in the corresponding figure
caption; the impulse map of Definition~\ref{def:impulse} transfers
$\Delta\,C(t_k^-)$ from the mainstream pool to the disengagement
compartment at each shock time.

Figure~\ref{fig:regimes} illustrates the three qualitatively distinct
post-shock trajectories of Section~\ref{sec:legacy}. The left panel
($\Delta=0.10<\Delta_c=0.25$) shows a small state shock: no initial
radical surge, convergence to the radical-free equilibrium
$C^\infty=1$. The centre panel ($\Delta=0.55>\Delta_c$, $\Delta\beta=0$)
shows a large pure state shock: transient radical growth followed by
full long-run recovery to $C^\infty=1$, with the disengaged pool $A$
(dotted) depleted at rate at least $\rho$ (baseline timescale
$1/\rho=10$). The right panel ($\Delta=0.55$, $\Delta\beta=0.15$, so
$\beta'=0.45>\mu$) shows a mixed state and structural shock: the system
converges to a new equilibrium with $C^\infty=\mu/\beta'\approx0.89$.
Conditional on a nonzero radical seed, the state component determines
the transient path, while the structural component determines the
long-run mainstream share. Here $P=L=R$ is plotted per radical wing, so
total radical support is $S=2P$.

\begin{figure}[ht]
\centering
\includegraphics[width=\textwidth]{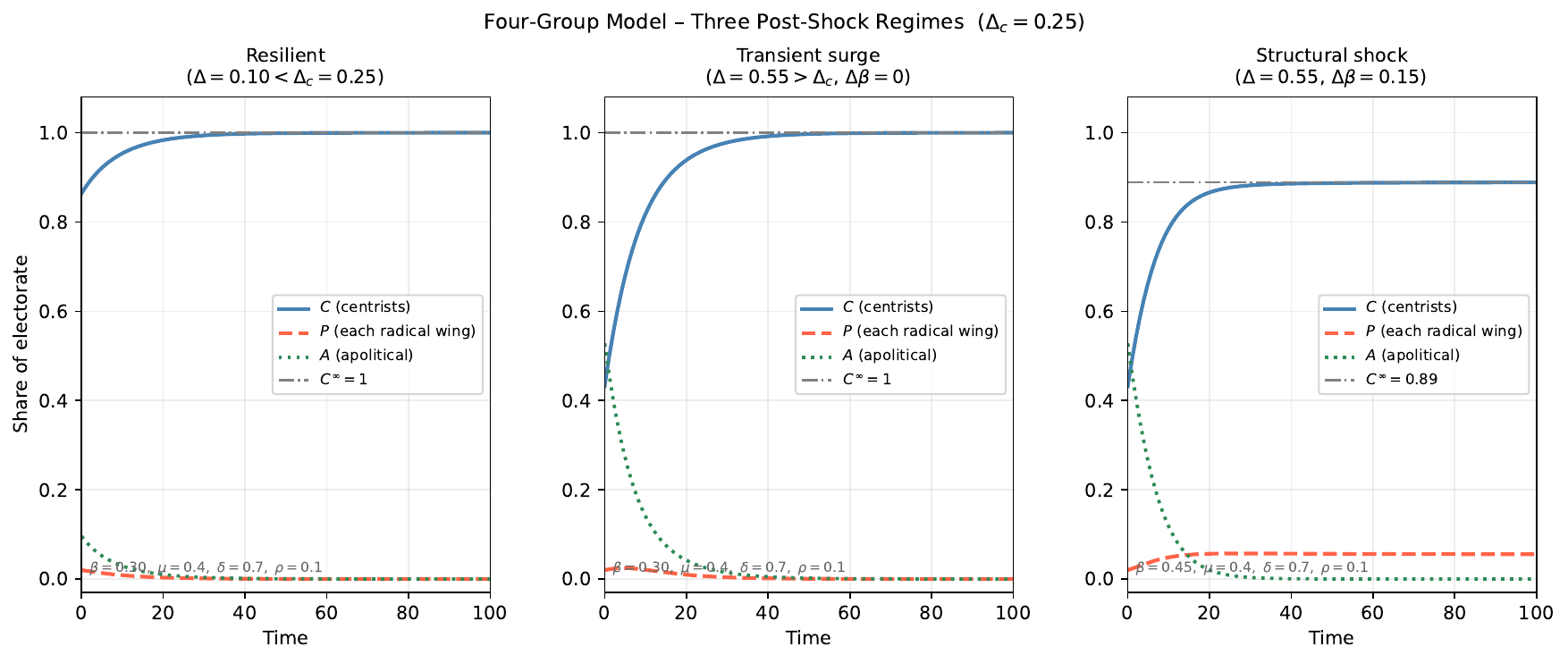}
\caption{Three shock regimes in the symmetric extended model.
  Parameters: $\beta_0=0.30$, $\mu=0.40$, $\delta=0.70$, $\rho=0.10$
  ($\Delta_c=0.25$).}
\label{fig:regimes}
\end{figure}

Figure~\ref{fig:staircase} simulates four shocks at times
$t=10,25,40,60$, each with state component $\Delta_k=0.40$ and
structural component $\Delta\beta_k=0.04$. The cumulative structural
parameter $B_k$ crosses $\mu=0.40$ at $k^*=3$. Before $k^*$, shocks are
transient in the long run: in the absence of further shocks the
mainstream share would return to $C^\infty=1$, although in the finite
simulation it only partially recovers between shocks. After $k^*$, the
long-run mainstream share is $C_k^\infty=\mu/B_k<1$ and decreases
whenever $\Delta\beta_k>0$; after the fourth shock, $B_4=0.46$ and
$C_4^\infty\approx0.870$. This is the staircase pattern of
Theorem~\ref{thm:cumulative} and Corollary~\ref{cor:staircase}.

\begin{figure}[ht]
\centering
\includegraphics[width=0.85\textwidth]{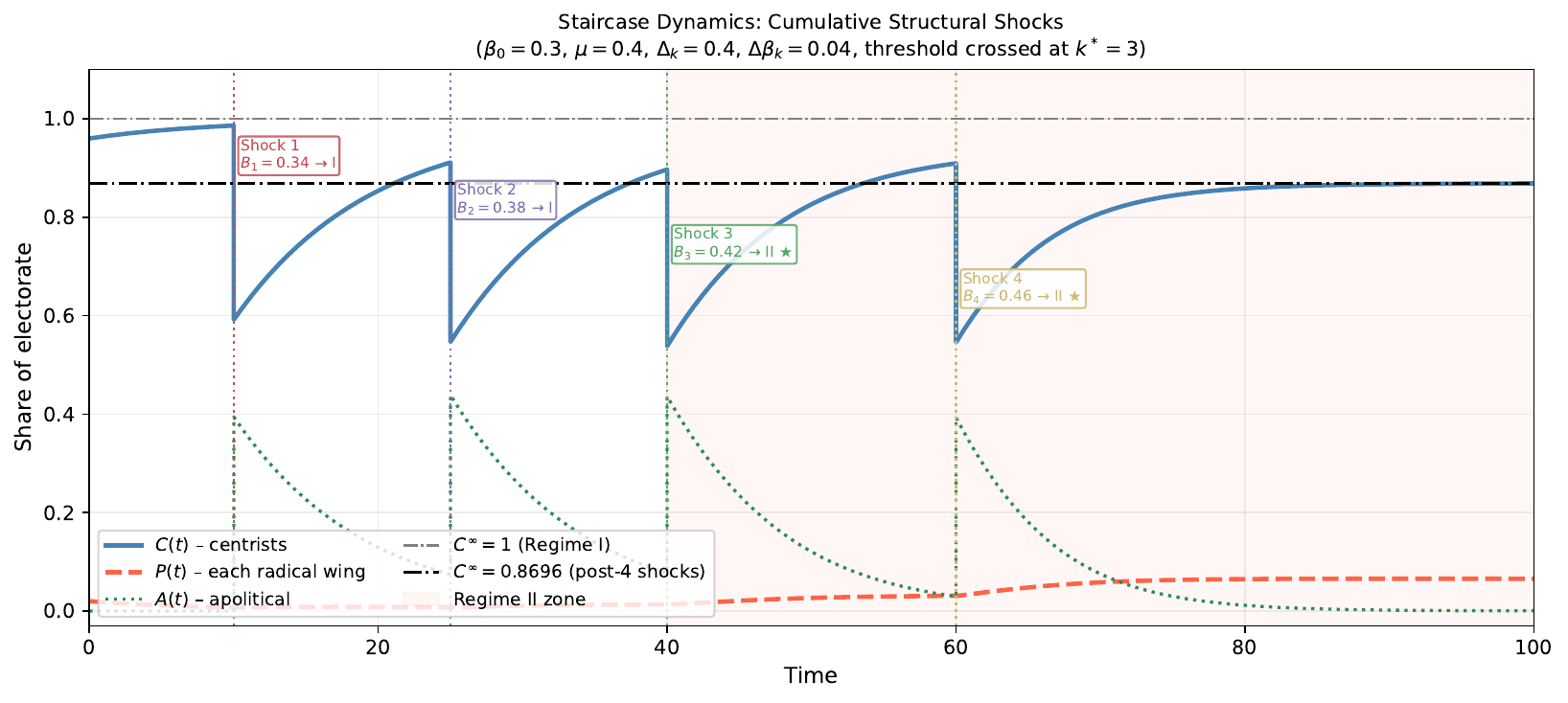}
\caption{Staircase dynamics under four sequential shocks
  ($\Delta_k=0.40$, $\Delta\beta_k=0.04$; $k^*=3$). Parameters:
  $\beta_0=0.30$, $\mu=0.40$, $\delta=0.70$, $\rho=0.10$.}
\label{fig:staircase}
\end{figure}

Figure~\ref{fig:asym_traj} shows trajectories of the full asymmetric
system~\eqref{eq:4group} for two parameter sets, both with an impulse
state shock of amplitude $\Delta=0.35$ at $t=5$. The left panel
($\mathcal{R}_{\mathrm{rad}}=0.840<1$) is subcritical: the shock creates
a temporary disengagement spike and a temporary drop in the mainstream
share, but the system returns to the radical-free equilibrium
$C^\infty=1$. The right panel ($\mathcal{R}_{\mathrm{rad}}=1.593>1$) is
supercritical: the system converges to the asymmetric radical-support
equilibrium $L^*=0.285$, $R^*=0.088$, $C^*=0.628$, consistent with
Theorem~\ref{thm:global4}.

\begin{figure}[ht]
\centering
\includegraphics[width=\textwidth]{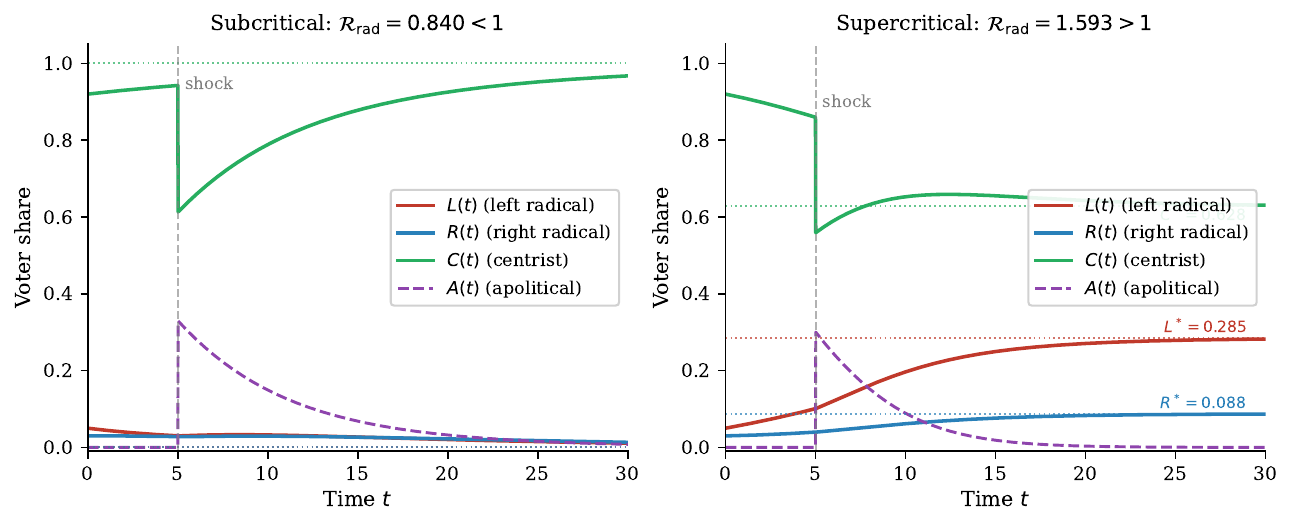}
\caption{Four-group asymmetric trajectories with impulse state shock
  $\Delta=0.35$ at $t=5$.
  Left: $\alpha_L=0.15$, $\alpha_R=0.20$, $\gamma_{RL}=0.08$,
  $\gamma_{LR}=0.12$, $\delta_L=0.70$, $\delta_R=0.55$, $\mu_L=0.30$,
  $\mu_R=0.35$, $\rho=0.12$.
  Right: $\alpha_L=0.40$, $\alpha_R=0.25$, $\gamma_{RL}=0.15$,
  $\gamma_{LR}=0.08$, $\delta_L=0.75$, $\delta_R=0.60$, $\mu_L=0.28$,
  $\mu_R=0.32$, $\rho=0.10$.}
\label{fig:asym_traj}
\end{figure}

Figure~\ref{fig:phi_delta} plots
$\Phi(\Delta):=\lambda_{\mathrm{PF}}(M^{-1}((1-\Delta)K+\Delta D))$ for
symmetric and asymmetric parameter sets. Both satisfy $\Phi(0)<1$ and
$\Phi(1)>1$, and each crosses $\Phi=1$ once in the plotted regime, at
$\Delta_c^{\mathrm{sym}}=0.091$ and $\Delta_c^{\mathrm{asym}}=0.116$
(vertical dotted lines). The figure illustrates the critical-threshold
construction in Theorem~\ref{thm:PF_Delta}.

\begin{figure}[ht]
\centering
\includegraphics[width=0.65\textwidth]{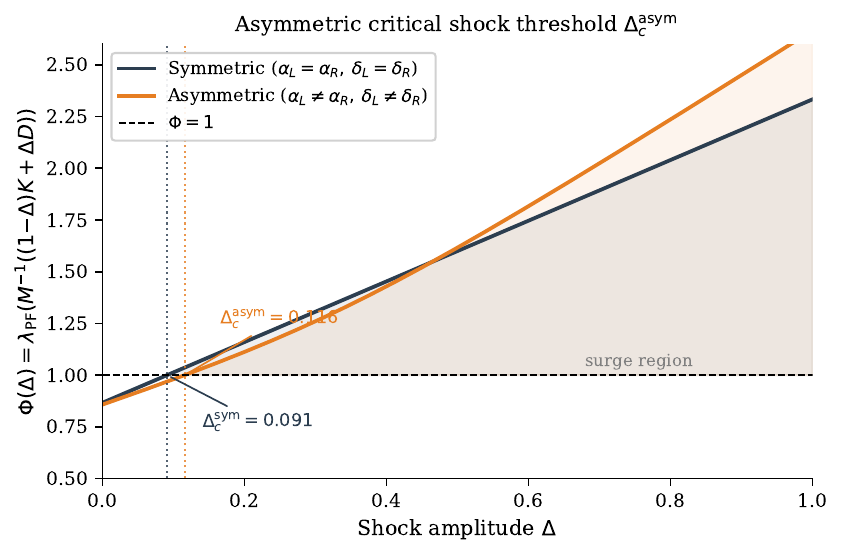}
\caption{$\Phi(\Delta)$ for symmetric and asymmetric parameter sets.
  Symmetric: $\alpha=0.18$, $\gamma=0.08$, $\mu=0.30$, $\delta=0.70$.
  Asymmetric: $\alpha_L=0.15$, $\alpha_R=0.22$, $\gamma_{RL}=0.08$,
  $\gamma_{LR}=0.05$, $\mu_L=\mu_R=0.30$, $\delta_L=0.80$,
  $\delta_R=0.55$.}
\label{fig:phi_delta}
\end{figure}

\bibliography{mss_paper}

\end{document}